\documentclass{article}

\usepackage[preprint,nonatbib]{neurips_2020}

\usepackage[utf8]{inputenc} % allow utf-8 input
\usepackage[T1]{fontenc}    % use 8-bit T1 fonts
\usepackage{url}            % simple URL typesetting
\usepackage{booktabs}       % professional-quality tables
\usepackage{amsfonts}       % blackboard math symbols
\usepackage{nicefrac}       % compact symbols for 1/2, etc.
\usepackage{microtype}      % microtypography

\usepackage{amsthm}
\usepackage{graphicx}
\usepackage{enumitem}
\usepackage{nicefrac}

\usepackage{xcolor}

\newtheorem{theorem}{Theorem}
\newtheorem{lem}[theorem]{Lemma}
\newtheorem{cor}[theorem]{Corollary}
\newtheorem{definition}[theorem]{Definition}
\newtheorem{prop}[theorem]{Proposition}
\newtheorem{claim}[theorem]{Claim}
\usepackage{thmtools} 
\usepackage{thm-restate}
\usepackage{algorithm}

\usepackage{amssymb}
\usepackage{amsmath}
\newif{\ifappendix}
\appendixtrue

\newcommand{\leaves}{\mathrm{leaves}}

\newcommand{\val}{\mathrm{val}}
\newcommand{\rev}{\mathrm{rev}}
\newcommand{\cost}{\mathrm{cost}}

\newcommand{\unfair}{\mathrm{unfair}}

\newcommand{\simi}{s}
\newcommand{\vweight}{m}
\newcommand{\dist}{d}
\newcommand{\fsize}{m_f}

\newcommand{\costl}{\ell}
\newcommand{\costm}{t}

\newcommand{\cy}{\mathcal{Y}}
\newcommand{\eps}{\epsilon}

\newcommand{\totr}{r_t}
\newcommand{\totb}{b_t}
\newcommand{\hc}{hierarchical clustering }

\newcommand{\weight}{m}
\newenvironment{proofof}[1]{\smallskip\noindent{\bf Proof of #1}}%
        {\hspace*{\fill}$\qed$\par}

\newcommand{\marina}[1]{\textcolor{red}{Marina: #1}}
\newcommand{\yuyan}[1]{\textcolor{blue}{Yuyan: #1}}

\usepackage{balance}
\usepackage{xspace}
\usepackage{algorithmic}

\newif\ifappendix
\appendixtrue
\newcommand{\refappendix}[1]{\ifappendix
 Appendix~\ref{#1}\xspace
\else
 the supplementary material\xspace
\fi}

\newcommand{\citet}{\cite}
\newcommand{\citep}{\cite}
\newcommand{\para}[1]{\smallskip\noindent{\bf #1.}\/}

\title{Fair Hierarchical Clustering}

\author{%
  Sara Ahmadian \\
  Google\\
  \texttt{sahmadian@google.com} \\
  \And
  Alessandro Epasto \\
  Google \\
  \texttt{aepasto@google.com} \\
  \AND
  Marina Knittel \\
 University of Maryland\\
  \texttt{mknittel@cs.umd.edu} \\
  \And
  Ravi Kumar \\
  Google \\
  \texttt{tintin@google.com} \\
  \And
  Mohammad Mahdian  \\
  Google \\
  \texttt{mahdian@google.com} \\
  \And
  Benjamin Moseley \\
  Carnegie Mellon University \\
  \texttt{moseleyb@andrew.cmu.edu} \\
  \And
  Philip Pham \\
  Google \\ 
  \texttt{phillypham@google.com}\\
  \And
  Sergei Vassilvtiskii \\
  Google \\
  \texttt{sergeiv@google.com} \\
  \And
  Yuyan Wang \\
  Carnegie Mellon University \\
  \texttt{yuyanw@andrew.cmu.edu}
}

\begin{document}

\maketitle

\begin{abstract}
As machine learning has become more prevalent, researchers have begun to recognize the necessity of ensuring machine learning systems are fair. Recently, there has been an interest in defining a notion of fairness that mitigates over-representation in traditional clustering.

In this paper we extend this notion to  hierarchical clustering, where the goal is to recursively partition the data to optimize a specific objective.  For various natural objectives, we obtain simple, efficient algorithms to find a provably good fair hierarchical clustering.  Empirically, we show that our algorithms can find a fair hierarchical clustering, with only a negligible loss in the objective.  

\end{abstract}

\section{Introduction}
Algorithms and machine learned models are increasingly used to assist in decision making on a wide range of issues, from mortgage approval to court sentencing recommendations~\cite{kleinberg2017human}. It is clearly undesirable, and in many cases illegal, for models to be biased to groups, for instance to discriminate on the basis of race or religion. Ensuring that there is no bias is not as easy as removing these protected categories from the data. Even without them being explicitly listed, the correlation between sensitive features and the rest of the training data may still cause the algorithm to be biased. This has led to an emergent literature on computing provably fair outcomes (see the book~\citet{hardtbook}).

The prominence of clustering in data analysis, combined with its use for data segmentation, feature engineering, and visualization makes it  critical that efficient fair clustering methods are developed. There has been a flurry of recent results in the ML research community, proposing algorithms for fair {\em flat} clustering, i.e., partitioning a dataset into a set of disjoint clusters, as captured by \textsc{k-center}, \textsc{k-median}, \textsc{k-means}, correlation clustering objectives~\citep{ahmadian,ahmadian2020fair,backurs, bera2019fair,bercea2018cost, chen2019proportionally,chiplunkar2020fair,huang2019coresets,jones2020fair,kleindessner2, spectralfair}. However, the same issues affect {\em hierarchical clustering}, which is the problem we study. 

The input to the \hc problem is a set of data points, with pairwise similarity or dissimilarity scores. A \hc is a tree, whose leaves correspond to the individual datapoints. Each internal node represents a cluster containing all the points in the leaves of its subtree. Naturally, the cluster at an internal node is the union of the clusters given by its children. Hierarchical clustering is widely used in data analysis~\cite{dubes1980clustering}, social  networks~\cite{mann2008use,rajaraman2011mining}, image/text organization~\cite{karypis2000comparison}. 

Hierarchical clustering is frequently used for flat clustering when the number of clusters is unknown ahead of time. A hierarchical clustering yields a set of clusterings at different granularities that are consistent with each other. Therefore, in all clustering problems where fairness is desired but the number of clusters is unknown, fair hierarchical clustering is useful.  As concrete examples, consider a collection of news articles organized by a topic hierarchy, where we wish to ensure that no single source or view point is over-represented in a cluster; or a hierarchical division of a geographic area, where the sensitive attribute is gender or race, and we wish to ensure balance in every level of the hierarchy.  There are many such  problems that benefit from fair hierarchical clustering, motivating the study of the problem area.

\para{Our contributions} We initiate an algorithmic study of fair hierarchical clustering.
%concurrently with Chhabra et al.~\cite{chhabra2020fair}, who explore a much different problem in this topic. 
We build on Dasgupta's seminal formal treatment of  hierarchical clustering~\citet{dasgupta} and prove our results for the revenue~\cite{moseleywang}, value~\cite{cohenaddad}, and cost~\cite{dasgupta} objectives in his framework.  %We begin by formally stating the problem.   It is easy to see that requiring {\em all} internal nodes to be balanced would be infeasible, as the nodes near the leaves may contain too few points to admit balance.  Our first contribution lies in defining the notion of balanced hierarchical clusterings.  

To achieve fairness, we show how to extend the  fairlets machinery, introduced by~\citet{chierichetti} and extended by~\citet{ahmadian}, to this problem.  We then investigate the  complexity of finding a good fairlet decomposition, giving both strong computational lower bounds and polynomial time approximation algorithms. 

%This requires first grouping all of the data points into minimal sets that retain balance, i.e. {\em fairlets}, and then running the hierarchical clustering algorithm on the fairlets. There is an exponential number of ways to decompose points into sets of fairlets and we define the best possible decomposition for hierarchical clustering objectives. We first show that finding the optimal decomposition is NP-hard (Theorem \ref{thm:np_hard_constant_approx}). Then we show that there is a polynomial time  algorithm that finds a fairlet decomposition, which gives a constant approximation when used in concert with known hierarchical clustering algorithms (Theorem \ref{thm:reduction}).  Specifically we show that our approach works for both the \hc objectives introduced by ~\citet{cohen2019hierarchical} (Corollary \ref{cor:main}) and ~\citet{DBLP:conf/nips/MoseleyW17} (Theorem \ref{thm:fair_moseley_wang}).

Finally, we conclude with an empirical evaluation of our approach. We show that ignoring protected attributes when performing \hc can lead to unfair clusters. On the other hand, adopting the fairlet framework in conjunction with the approximation algorithms we propose yields {\em fair} clusters with a {\em negligible} objective degradation.

\para{Related work} 
Hierarchical clustering has received increased attention over the past few years. Dasgupta~\citet{dasgupta} developed a cost function objective for data sets with similarity scores, where similar points are encouraged to be clustered together lower in the tree. Cohen-Addad et al.~\citet{cohenaddad} generalized these results into a class of optimization functions that possess other desirable properties and introduced their own value objective in the dissimilarity score context. %The objective they proposed, which we refer to as \ckmm, is an analogy of Dasgupta's objective for points with \emph{dissimilarity} scores. They also defined \emph{ground-truth inputs} for the \hc problem, and showed that optimizing the \ckmm\ objectives gives desirable optimal solutions for the ground-truth inputs. 
In addition to validating their objective on inputs with known ground truth, they gave a theoretical justification for the average-linkage algorithm, one of the most popular algorithms used in practice, as a constant-factor approximation for value. %, showing that it is a $\frac{1}{2}$-approximation for \ckmm\ objective. They then designed a different algorithm that gives a $\frac{2}{3}-\epsilon$ approximation for this objective.
Contemporaneously, Moseley and Wang \citep{moseleywang} designed a revenue objective function based on  the work of Dasgupta for point sets with similarity scores and showed the average-linkage algorithm is a constant approximation for this objective as well. This work was further improved by Charikar \citet{charikar19} who gave a tighter analysis of average-linkage  for Euclidean data for this objective and \cite{AhmadianBisect,YossiNogaHC} who improved the approximation ratio in the general case. 

In parallel to the new developments in algorithms for hierarchical clustering, there has been tremendous development in the area of fair machine learning. We refer the reader to a recent textbook ~\citep{hardtbook} for a rich overview, and focus here on progress for fair clustering.  Chierichetti et al.~\citet{chierichetti} first defined fairness for $k$-median and $k$-center clustering, and introduced the notion of {\em fairlets} to design efficient algorithms. Extensive research has focused on two topics:  adapting the definition of fairness to broader contexts, and designing efficient algorithms for finding good fairlet decompositions. For the first topic, the fairness definition was extended to multiple values for the protected feature~\citet{ahmadian, bercea2018cost,rosner}.  For the second topic, Backurs et al.~\citet{backurs} proposed a near-linear constant approximation algorithm for finding fairlets for $k$-median, Kleindessner et al.~\citet{kleindessner2} designed a linear time constant approximation algorithm for $k$-center, Bercea et al.~\citet{bercea2018cost} developed methods for fair $k$-means, while Ahmadian et al.~\cite{ahmadian2020fair} defined approximation algorithms for fair correlation clustering. Concurrently with our work, Chhabra et al.~\cite{chhabra2020fair} introduced a possible approach to ensuring fairness in hierarchical clustering. However, their fairness definition differs from ours (in particular, they do not ensure that all levels of the tree are fair), and the methods they introduce are heuristic, without formal fairness or quality guarantees. 

%Concurrently to our work, Chhabra et al.~\cite{chhabra2020fair} introduced a different approach in fairness for hierarchical clustering. Instead of enforcing fairness constraint, as in our paper, they define a fairness cost which is a function of the overall imbalance of the sub-cluster in the hierarchy. More precisely, their notion measures how different the distribution of the protected group in a cluster is from the overall data distribution. Their work does not address optimizing any formal hierarchical clustering objective, as our work does, but instead provides a heuristic to reduce the fairness cost while performing standard linkage-based hierarchical clustering. Given the different problem definition, their approach does not provide a solution to our problem. 

Beyond clustering, the same balance notion that we use has been utilized to capture fairness in other contexts, for instance: fair voting~\citep{voting}, fair optimization~\citep{matroids}, as well as other problems~\citep{Ranking}. 
\vspace{-0.1in}
\section{Formulation}

\subsection{Objectives for hierarchical clustering}

Let $G = (V, s)$ be an input instance, where $V$ is a set of $n$ data points, and $s:V^2 \rightarrow \mathbb{R}^{\geq 0}$ is a similarity function over vertex pairs. For two sets, $A, B \subseteq V$, we let $s(A, B) = \sum_{a \in A, b\in B} s(a, b)$ and $S(A) = \sum_{\{i,j\} \in A^2}s(i,j)$. For problems where the input is  $G = (V, d)$, with $d$ a distance function, we define $d(A,B)$ and $d(A)$ similarly. 
We also consider the \emph{vertex-weighted} versions of the problem, i.e. $G = (V, s, \vweight)$ (or $G = (V, d, \vweight)$), where $\vweight : V \rightarrow \mathbb{Z}^{+}$ is a weight function on the vertices.  The vertex-unweighted version can be interpreted as setting $\vweight(i) = 1, \forall i \in V$.  For $U \subseteq V$, we use the notation $\vweight(U) = \sum_{i \in U} \vweight(i)$.

A \emph{hierarchical clustering} of $G$ is a tree whose leaves correspond to $V$ and whose internal vertices represent the merging of vertices (or clusters) into larger clusters until all data merges at the root.  The goal of hierarchical clustering is to build a tree to optimize some objective.  

To define these objectives formally, we need some notation.  Let $T$ be a hierarchical clustering tree of $G$.  For two leaves $i$ and $j$, we say $i\lor j$ is their least common ancestor.
For an internal vertex $u$ in $T$, let $T[u]$ be the subtree in $T$ rooted at $u$.  Let $\leaves(T[u])$ be the leaves of $T[u]$.%, and $\nonleaves(T[u]) = \leaves(T) \setminus \leaves(T[u])$.    

We consider three different objectives---{\em revenue}, {\em value}, and {\em cost}---based on the seminal framework of~\citet{dasgupta}, and generalize them to the vertex-weighted case.

\textbf{Revenue.}  Moseley and Wang~\citet{moseleywang} introduced the revenue objective for hierarchical clustering.  Here the input instance is of the form $G = (V, \simi, m)$, where $\simi: V^2 \rightarrow \mathbb{R}^{\geq 0}$ is a \emph{similarity} function.
%on pairs of vertices.
%
\begin{definition}[Revenue]
\label{def:rev}
The \emph{revenue} ($\rev$) of a tree $T$ for an instance $G = (V, \simi, \vweight)$, where $\simi(\cdot, \cdot)$ denotes similarity between data points, is:
 $\rev_G(T) = \sum_{i,j\in V} \simi(i,j) \cdot \big(\vweight(V) - \vweight(\leaves(T[i\lor j])) \big).$
\end{definition}
Note that in this definition, each weight is scaled by (the vertex-weight of) the non-leaves.  The goal is to find a tree of maximum revenue.  It is known that average-linkage is a $\nicefrac{1}{3}$-approximation for vertex-unweighted revenue~\cite{moseleywang}; the state-of-the-art is a $0.585$-approximation~\cite{YossiNogaHC}.

%For $U \subseteq V$, let $\simi(U) = \sum_{i, j \in U} \simi(i, j)$.  
As part of the analysis, there is an upper
bound for the revenue objective~\citet{cohenaddad,moseleywang}, which is easily extended to the vertex-weighted setting:
\begin{equation}
    \label{eq:maxrevenue}
    \rev_G(T) \leq \Big( \vweight(V) - \min_{u, v \in V, u \neq v} \vweight(\{u, v\}) \Big) \cdot \simi(V).
\end{equation}
Note that in the vertex-unweighted case, the upper bound is just $(|V| - 2) s(V)$.

\textbf{Value.}
A different objective was proposed by Cohen-Addad et al.~\citet{cohenaddad}, using distances instead of similarities.   Let $G=(V,\dist,m)$, where $\dist:V^2\to\mathbb{R}^{\geq0}$ is a distance (or dissimilarity) function.
\begin{definition}[Value]
\label{def:val}
The \emph{value} ($\val$) of a tree $T$ for an instance $G = (V, \dist, \vweight)$ where $\dist(\cdot, \cdot)$ denotes distance is: $\val_G(T) = \sum_{i,j\in V}\dist(i,j) \cdot \vweight(\leaves(T[i\lor j])).$
\end{definition}
As in revenue, we aim to find a hierarchical clustering to maximize value.
Cohen-Addad et al.~\citet{cohenaddad} showed that both average-linkage and a locally $\epsilon$-densest cut algorithm achieve a $\nicefrac{2}{3}$-approximation for vertex-unweighted value.   They also provided an upper bound for value, much like that in \eqref{eq:maxrevenue}, which in the vertex-weighted context, is:
\begin{equation}
    \label{eq:maxvalue}
    \val_G(T) \leq \vweight(V) \cdot \simi(V).
\end{equation}
\textbf{Cost.}
The original objective introduced by Dasgupta \citet{dasgupta} for analyzing hierarchical clustering algorithms introduces the notion of cost. 
\begin{definition}[Cost]
\label{def:cost}
The $\cost$ of a tree $T$ for an instance $G = (V, \simi)$ where $\simi(\cdot, \cdot)$ denotes similarity is: $\cost_G(T) = \sum_{i,j\in V} \simi(i,j) \cdot |\leaves(T[i\lor j])|.$

\end{definition}

The objective is to find a tree of minimum cost. From a complexity point of view, cost is a harder objective to optimize.
Charikar and Chatziafratis~\citet{charikar17} showed that cost is not constant-factor approximable under the Small Set Expansion hypothesis, and the current best approximations are $O\left(\sqrt{\log n}\right)$ and require solving SDPs.

{\bf Convention.} Throughout the paper we adopt the following convention: $\simi(\cdot, \cdot)$ will always denote similarities and $\dist(\cdot, \cdot)$ will always denote distances.  Thus, the inputs for the cost and revenue objectives will be instances of the form $(V, \simi, \vweight)$ and inputs for the value objective will be instances of the form $(V, \dist, \vweight)$.  All the missing proofs can be found in the Supplementary Material. 
\subsection{Notions of fairness}

Many definitions have been proposed for fairness in clustering.  We consider the setting in which each data point in $V$ has a \emph{color}; the color corresponds to the protected attribute.  

\textit{Disparate impact.} This notion is used to capture the fact that decisions (i.e., clusterings) should not be overly favorable to one group versus another. 
This notion was formalized by Chierichetti et al.~\citet{chierichetti} for clustering  when the protected attribute can take on one of two values, i.e., points have one of two colors. 
In their setup, the \emph{balance} of a cluster is the ratio of the minimum to the maximum number of points of any color in the cluster.  Given a balance requirement $t$, a clustering is fair if and only if each cluster has a balance of at least $t$.

\textit{Bounded representation.} A generalization of disparate impact, bounded representation focuses on mitigating the imbalance of the representation of protected classes (i.e.,  colors) in clusters and was defined by Ahmadian et al.~\citet{ahmadian}.  Given an over-representation parameter $\alpha$, a cluster is fair if the \emph{fractional representation} of each color in the cluster is at most $\alpha$, and a clustering is fair if each cluster has this property.
An interesting special case of this notion is when there are $c$ total colors and $\alpha=1/c$. In this case, we require that every color is equally represented in every cluster. We will refer to this as \emph{equal representation}. These notions enjoy the following useful property:
\begin{definition}[Union-closed]
A fairness constraint is \emph{union-closed} if for any pair of fair clusters $A$ and $B$, $A\cup B$ is also fair.
\end{definition}
This property is useful in hierarchical clustering: given a tree $T$ and internal node $u$, if each child cluster of $u$ is fair, then $u$ must also be a fair cluster.  %We can now formally extend fair clustering to hierarchical clustering.

\begin{definition}[Fair hierarchical clustering]
For any fairness constraint, a \emph{hierarchical clustering is fair} if all of its clusters (besides the leaves) are fair.
%, and the minimal non-leaf clusters form a clustering of all the data.
% Sergei: removed above after consulting with epasto
\end{definition}
Thus, under any union-closed fairness constraint, this definition is equivalent to restricting the bottom-most clustering (besides the leaves) to be fair. Then given an objective (e.g., revenue), the goal is to find a fair hierarchical clustering that optimizes the objective.  We focus on the bounded representation fairness notion with $c$ colors and an over-representation cap $\alpha$. However, the main ideas for the revenue and value objectives work under any notion of fairness that is union-closed.

\section{Fairlet decomposition}

\begin{definition}[Fairlet~\citet{chierichetti}]
A \emph{fairlet} $Y$ is a fair set of points such that there is no partition of $Y$ into $Y_1$ and $Y_2$ with both $Y_1$ and $Y_2$ being fair. 
\end{definition}
In the bounded representation fairness setting, a set of points is fair if at most an $\alpha$ fraction of the points have the same color.  We call this an \emph{$\alpha$-capped fairlet}.  For $\alpha=1/t$ with $t$ an integer, the fairlet size will always be at most $2t-1$. We will refer to the maximum size of a fairlet by $m_f$. 

Recall that given a union-closed fairness constraint, if the bottom clustering in the tree is a layer of fairlets (which we call a \textit{fairlet decomposition} of the original dataset) the hierarchical clustering tree is also fair.  This observation gives an immediate algorithm for finding fair hierarchical clustering trees in a two-phase manner.  (i) Find a fairlet decomposition, i.e., partition  the input set $V$ into clusters $Y_1, Y_2, \ldots$ that are all fairlets.  (ii) Build a tree on top of all the fairlets.  Our goal is to complete both phases in such a way that we optimize the given objective (i.e., revenue or value). 

In Section~\ref{sec:rev}, we will see that to optimize for the revenue objective, all we need is a fairlet decomposition with bounded fairlet size. However, the fairlet decomposition required for the value objective is more nuanced.  We describe this next.

\paragraph{Fairlet decomposition for the value objective}

For the value objective, we need the total distance between pairs of points inside each fairlet to be small.  Formally, suppose $V$ is partitioned into fairlets $\mathcal{Y} = \{Y_1, Y_2, \ldots\}$ such that $Y_i$ is an $\alpha$-capped fairlet.  The cost of this decomposition is defined as:
\begin{equation}
\label{obj:fairlet_obj}
 \phi(\mathcal{Y}) = \sum_{Y \in \mathcal{Y}}\sum_{\{u,v\} \subseteq Y}\dist(u,v).
\end{equation}
Unfortunately, the problem of finding a fairlet decomposition to minimize $\phi(\cdot)$ does not admit any constant-factor approximation unless P = NP. 
\begin{theorem} \label{thm:np_hard_constant_approx}
Let $z \geq 3$ be an integer. Then 
%For any constants $(b,r)$ where $0 \leq b \leq r$, $\gcd(b,r)=1$ and $r \geq 3$, 
there is no bounded approximation algorithm for finding $(\frac{z}{z+1})$-capped fairlets optimizing $\phi(\mathcal{Y})$, which runs in polynomial time, unless P = NP.
\end{theorem}
%Note that this hardness result is sufficient to prove the hardness for the multi-colored version.
The proof proceeds by a reduction from the Triangle Partition problem, which asks if a graph $G = (V, E)$ on $3n$ vertices can be partitioned into three element sets, with each set forming a triangle in $G$. 
%We defer the full proof to \ifappendix Appendix~\ref{app:hardness}.  \else the  Supplementary Material. \fi
Fortunately, for the purpose of optimizing the value objective, it is not necessary to find an approximate decomposition. 
%We will see in Section~\ref{sec:val} that we simply need  to capture a sufficiently small fraction of the total weight in the edge weights between clusters.

\section{Optimizing revenue with fairness}\label{sec:rev}

This section  considers the revenue objective.  We will obtain an approximation algorithm for this objective in three steps: (i) obtain a fairlet decomposition such that the maximum fairlet size in the decomposition is small, (ii) show that any  $\beta$-approximation algorithm to (\ref{eq:maxrevenue}) plus this fairlet decomposition can be used to obtain a (roughly) $\beta$-approximation for fair hierarchical clustering under the revenue objective, and (iii) use average-linkage, which is known to be a $\nicefrac{1}{3}$-approximation to (\ref{eq:maxrevenue}).  (We note that the recent work~\cite{AhmadianBisect,YossiNogaHC} on improved approximation algorithms compare to a bound on the optimal solution that differs from (\ref{eq:maxrevenue}) and therefore do not fit into our framework.)

First, we address step (ii).  Due to space, this proof can be found in~\refappendix{sec:thmproofsec}.
\begin{theorem} \label{thm:fair_moseley_wang}
Given an algorithm that obtains a $\beta$-approximation to (\ref{eq:maxrevenue}) where $\beta \leq 1$, and a fairlet decomposition with maximum fairlet size $m_f$, there is a $\beta \left(1-\frac{2\fsize}{n}\right)$-approximation for fair hierarchical clustering under the revenue objective.
\end{theorem}

Prior work showed that average-linkage is a $\nicefrac{1}{3}$-approximation to (\ref{eq:maxrevenue}) in the vertex-unweighted case; this proof can be easily modified to show that it is still $\nicefrac{1}{3}$-approximation even with vertex weights. This accounts for step (iii) in our process.
%\ifappendix (Section \ref{app:weighted_algo}). \else (See supplementary material).\fi

Combined with the fairlet decomposition methods for the two-color case~\citep{chierichetti} and for multi-color case (Supplementary Material) to address step (i),
%\ifappendix (Section~\ref{app:multi_color}), \else the  supplementary material, \fi 
we have the following.

%\begin{corollary}
\begin{cor} \label{cor:revenue}
There is polynomial time algorithm that constructs a fair tree that is a $\frac{1}{3}\left(1-\frac{2\fsize}{n}\right)$-approximation for revenue objective, where $\fsize$ is the maximum size of fairlets.
\end{cor}

\section{Optimizing value with fairness}\label{sec:val}

In this section we consider the value objective.  As in the revenue objective, we prove that we can reduce fair hierarchical clustering to the problem of finding a good fairlet decomposition for the proposed fairlet objective \eqref{obj:fairlet_obj}, and then use any approximation algorithm for weighted hierarchical clustering with the decomposition as the input.

\begin{theorem}
\label{thm:reduction}
Given an algorithm that gives a $\beta$-approximation to \eqref{eq:maxvalue} where $\beta\leq 1$, and a fairlet decomposition $\mathcal{Y}$ such that $\phi(\mathcal{Y}) \leq \epsilon \cdot \dist(V)$, there is a $\beta(1 - \epsilon)$ approximation for \eqref{eq:maxvalue}.
\end{theorem} 
We complement this result with an algorithm that finds a good fairlet decomposition in polynomial time under the bounded representation fairness constraint with cap $\alpha$.

Let $R_1,\ldots,R_c$ be the $c$ colors and let $\cy = \{Y_1, Y_2 \ldots\}$ be the fairlet decomposition. 
%Let $\kfsize{k}$ be the size of the $k$th fairlet. 
Let $n_i$ be the number of points colored $R_i$ in $V$.  Let $r_{i, k}$ denote the number of points colored $R_i$ in the $k$th fairlet. 
\begin{theorem}
\label{thm:localsearch}
There exists a local search algorithm that finds a fairlet decomposition $\mathcal{Y}$ with $\phi(\mathcal{Y}) \leq (1 +\epsilon) \max_{i,k} \frac{r_{i,k}}{n_i} \dist(V)$ in time $\tilde{O}(n^3 / \epsilon)$. 
\end{theorem}
We can now use the fact that both average-linkage and the $\frac{\epsilon}{n}$-locally-densest cut algorithm give a $\frac{2}{3}$- and $(\frac{2}{3}-\epsilon)$-approximation respectively for vertex-weighted hierarchical clustering under the value objective.  Finally, recall that fairlets are intended to be minimal, and their size depends only on the parameter $\alpha$, and not on the size of the original input. Therefore, as long as the number of points of each color increases as input size, $n$, grows, the ratio $r_{i,k}/n_i$ goes to $0$. These results, combined with Theorem \ref{thm:reduction} and Theorem \ref{thm:localsearch}, yield Corollary \ref{cor:main}.

\begin{cor} \label{cor:main}
Given bounded size fairlets, the fairlet decomposition computed by local search combined with average-linkage constructs a fair hierarchical clustering that is a $\frac{2}{3}(1-o(1))$-approximation for the value objective. For the $\frac{\epsilon}{n}$-locally-densest cut algorithm in \citet{cohenaddad}, we get a polynomial time algorithm for fair hierarchical clustering that is a $(\frac{2}{3}-\epsilon)(1-o(1))$-approximation under the value objective for any $\epsilon > 0$.
\end{cor}

Given at most a small fraction of every color is in any cluster, Corollary \ref{cor:main} states that we can extend the state-of-the-art results for value to the $\alpha$-capped, multi-colored constraint. Note that the preconditions will always be satisfied and the extension will hold in the two-color fairness setting or in the multi-colored equal representation fairness setting.

\paragraph{Fairlet decompositions via local search}

In this section, we give a local search algorithm to construct a fairlet decomposition, which proves Theorem \ref{thm:localsearch}. This is inspired by the $\epsilon$-densest cut algorithm of \cite{cohenaddad}. To start, recall that for a pair of sets $A$ and $B$ we denote by $\dist(A,B)$ the sum of interpoint distances, $\dist(A,B)=\sum_{u \in A, v \in B}\dist(u,v)$. A fairlet decomposition is  a partition of the input $ \{Y_1, Y_2, \ldots \}$ such that each color composes at most an $\alpha$ fraction of each $Y_i$.

Our algorithm will recursively subdivide the cluster of all data to construct a hierarchy by finding cuts. To search for a cut, we will use a \textit{swap} method.

\begin{definition}[Local optimality] \label{def:ep_opt_decomp}
Consider any fairlet decomposition $\mathcal{Y} = \{Y_1, Y_2, \ldots \}$ and $\epsilon >0$.   Define a \emph{swap} of $u \in Y_i$ and $v \in Y_j$ for $j\neq i$ as updating $Y_i$ to be $(Y_i \setminus \{u\}) \cup \{v\}$ and $Y_j$ to be $(Y_j \setminus \{v\}) \cup \{u\}$.  We say $\mathcal{Y} $ is \emph{$\epsilon$-locally-optimal} if any swap with $u,v$ of the same color reduces the objective value by less than a $(1+\epsilon)$ factor. 
\end{definition}

The algorithm constructs a $(\epsilon/n)$-locally optimal algorithm for fairlet decomposition, which  runs in $\tilde{O}(n^3/\epsilon)$ time.   Consider any given instance $(V, \dist)$.  Let $\dist_{\max}$ denote the maximum distance, $\fsize$ denote the maximum fairlet size, and $\Delta = \dist_{\max} \cdot \frac{\fsize}{n}$.  The algorithm begins with an arbitrary decomposition. Then it swaps pairs of monochromatic points until it terminates with a locally optimal solution.  By construction we have the following.
\begin{claim}
\label{claim:fairdecomp}
Algorithm \ref{alg:eps_local_opt_decomp}  finds a valid fairlet decomposition.
\end{claim}

We prove two things:  Algorithm~\ref{alg:eps_local_opt_decomp} optimizes the objective (\ref{obj:fairlet_obj}), and  has a small running time. The following lemma gives an upper bound on $\mathcal{Y}$'s performance for \eqref{obj:fairlet_obj} found by Algorithm~\ref{alg:eps_local_opt_decomp}.

\begin{algorithm}[t!]
 \renewcommand{\algorithmicrequire}{\textbf{Input: }}
	\renewcommand{\algorithmicensure}{\textbf{Output: }}
    \caption{Algorithm for $(\epsilon/n)$-locally-optimal fairlet decomposition.}
    \label{alg:eps_local_opt_decomp}
    \begin{algorithmic}[1]
    	\REQUIRE A set $V$ with distance function $\dist \geq 0$, parameter $\alpha$, small constant $\epsilon \in [0,1]$.
        \ENSURE An $\alpha$-capped fairlet decomposition $\mathcal{Y}$.
        \STATE Find $\dist_{\max}$,  $\Delta \leftarrow \frac{m_f}{n}\dist_{\max}$.
        \STATE Arbitrarily find an $\alpha$-capped fairlet decomposition $\{Y_1, Y_2, \ldots \}$ such that each partition has at most an $\alpha$ fraction of any color.
                \WHILE{$\exists u \in Y_i, v \in Y_j, i \neq j$ of the same color, such that for the decomposition $\mathcal{Y}'$ after swapping $u,v$, $\frac{\sum_{Y_k \in \mathcal{Y}}\dist(Y_k)}{\sum_{Y_k \in \mathcal{Y}'}\dist(Y_k)} \geq (1+\epsilon/n)$ {\bf{and}} $\sum_{Y_k \in \mathcal{Y}}\dist(Y_k) > \Delta$}
        \STATE Swap $u$ and $v$  by setting $Y_i \gets (Y_i \setminus \{u\}) \cup \{v\}$ and $Y_j \gets (Y_j \setminus \{v\}) \cup \{u\}$.
        \ENDWHILE
        %\RETURN $\mathcal{Y}$.
    \end{algorithmic}
\end{algorithm}

\begin{lem} 
%\begin{theorem}
\label{lem:algo_is_correct}
The fairlet decomposition $\mathcal{Y}$ computed by Algorithm  \ref{alg:eps_local_opt_decomp}  has an objective value for (\ref{obj:fairlet_obj}) of at most $(1+\epsilon)\max_{i,k}\frac{r_{i,k}}{n_i}\dist(V)$.
%\end{theorem}
\end{lem}

Finally we bound the running time.  The algorithm has much better performance in practice than its worst-case analysis would indicate. We will show this later in Section~\ref{sec:exp}.

\begin{lem}
\label{lem:algo_is_fast1}
The running time for Algorithm~\ref{alg:eps_local_opt_decomp} is $\tilde{O}(n^3/\epsilon)$.
\end{lem}

Together, Lemma~\ref{lem:algo_is_correct},  Lemma~\ref{lem:algo_is_fast1}, and Claim~\ref{claim:fairdecomp} prove Theorem~\ref{thm:localsearch}. This establishes that there is a  local search algorithm that can construct a good fairlet decomposition.

\iffalse
Similar to our theorem for revenue, these results generalize to any union-closed constraint, but they are limited to only extend Average Linkage. At this point in time, however, this is the best approximation we can expect to achieve for value.

\begin{theorem}
Given any fairlet decomposition algorithm that always produces a decomposition of cost at most $o(\sum_{e\in E(G)} w(e))$, we can design a tight $(2/3 - o(1))$-approximation algorithm for optimizing value with any union-closed fairness constraint.
\end{theorem}

In addition to this result, we also show the tightness of Average Linkage's approximation generally (ie, with or without fairness constraints).

The fairlet prerequisite in this case is a bit more difficult to achieve. We can, at the very least, achieve this for the following constraint:

\begin{cor}
Given any graph with two of vertices and $\alpha=1/2$, we can design a tight $(2/3 - o(1))$-approximation algorithm for optimizing value with any union-closed fairness constraint.
\end{cor}

\fi
%\input{value_copy.tex}
\section{Optimizing cost with fairness}\label{sec:cost}

This section  considers the cost objective of~\citet{dasgupta}.  Even without our fairness constraint, the difficulty of approximating cost is clear in its approximation hardness and the fact that all known solutions require an LP or SDP solver.   We obtain the result in Theorem~\ref{thm:cost}; extending this result to other fairness constraints, improving its bound, or even making the algorithm practical, are open questions.

\begin{theorem}\label{thm:cost}
    Consider the two-color case. Given a $\beta$-approximation for cost and a $\gamma_\costm$-approximation for minimum weighted bisection
    \footnote{
    The minimum weighted bisection problem is to find a partition of nodes into two equal-sized subsets so that the sum of the weights of the edges crossing the partition is minimized.
    }
    on input of size $\costm$, then for parameters $\costm$ and $\costl$ such that $n\geq \costm\costl$ and $n> \costl + 108\costm^2/\costl^2$, there is a fair  $O\left(\frac n\costm + \costm\costl  +\frac{n\costl\gamma_\costm}{\costm}+ \frac{n\costm\gamma_\costm}{\costl^2}\right)\beta$-approximation for $\cost(T^*_{\unfair})$.
\end{theorem}

With proper parameterization, we achieve an $O\left(n^{5/ 6}\log^{5/4} n \right)$-approximation. 
%The algorithm and analysis  are quite complicated. 
We defer our algorithm description, pseudocode, and proofs to the Supplementary Material. While our algorithm is not simple, it is an important (and non-obvious) step to show the existence of an approximation, which we hope will spur future work in this area.

%The intricacy of our solution suggests that the dichotomy between approximating cost and other optimization functions persists in the fair setting. 
\newcommand{\CensusGender}{\textsc{CensusGender}\xspace}
\newcommand{\CensusRace}{\textsc{CensusRace}\xspace}
\newcommand{\BankMarriage}{\textsc{BankMarriage}\xspace}
\newcommand{\BankAge}{\textsc{BankAge}\xspace}
\newcommand{\ratio}{\mathrm{ratio}}
\newcommand{\myvalue}{\mathrm{value}}
\newcommand{\fairness}{\mathrm{fairness}}

\section{Experiments}
\label{sec:exp}
\vspace{-.2cm}

This section validates our algorithms from Sections~\ref{sec:rev} and~\ref{sec:val} empirically. We adopt the disparate impact fairness constraint~\cite{chierichetti}; thus each point is either blue or red.  In particular, we would like to:
\begin{itemize}[nosep]
\item Show that running the standard average-linkage algorithm results in highly unfair solutions. 
\item Demonstrate that demanding fairness in hierarchical clustering incurs only a small loss in the hierarchical clustering objective.  
\item Show that our algorithms, including fairlet decomposition, are practical on real data. 

\end{itemize}
In \refappendix{sec:multicolorexp} we consider multiple colors and the same trends as the two color case occur.

\begin{table*}[t!] \vspace{-0.15in}
\centering 
\caption{ \vspace{-0.2in} Dataset description.  Here $(b, r)$ denotes the balance of the dataset.} \label{table:datasets}
\small
\begin{tabular}{r|lllll}
\hline
Name & Sample size & \# features & Protected feature & Color (blue, red) & $(b,r)$ \\
\hline
\CensusGender & 30162 & 6 & gender & (female, male) & $(1,3)$ \\
\CensusRace   & 30162 & 6 & race & (non-white, white) & $(1,7)$ \\
\BankMarriage & 45211 & 7 & marital status & (not married, married) & $(1,2)$ \\
\BankAge      & 45211 & 7 & age & ($<40$, $\geq 40$) & $(2,3)$ \\
\hline
\end{tabular}
\end{table*}
\vspace{-0.1in}
\begin{table*}[t!]\vspace{-0.2in}
\centering 
\caption{Impact of Algorithm~\ref{alg:eps_local_opt_decomp} on $\ratio_{\myvalue}$ in percentage (mean $\pm$ std. dev).}
\label{emp:ratiockmm}
\tiny
\begin{tabular}{r|llllllll}
\hline
Samples & 400 & 800 & 1600 & 3200 & 6400 & 12800\\
\hline
\CensusGender, initial & $88.17 \pm 0.76$ & $88.39 \pm 0.21$ & $88.27 \pm 0.40$ & $88.12 \pm 0.26$ & $88.00 \pm 0.10$ & $88.04 \pm 0.13$ \\
final  & $99.01 \pm 0.60$ & $99.09 \pm 0.58$ & $99.55 \pm 0.26 $ & $99.64 \pm 0.13 $ & $ 99.20 \pm 0.38 $ & $99.44 \pm 0.23$ \\
\hline
\CensusRace, initial  & $84.49 \pm 0.66$ & $85.01 \pm 0.31$ & $85.00 \pm 0.42$ & $84.88 \pm 0.43$ & $84.84 \pm 0.16$ & $84.89 \pm 0.20 $ \\
final & $99.50 \pm 0.20$ & $99.89 \pm 0.32$ & $100.0 \pm 0.21$ & $99.98 \pm 0.21$ & $99.98 \pm 0.11$ & $99.93 \pm 0.31$ \\
\hline
\BankMarriage, initial  & $92.47 \pm 0.54$ & $92.58 \pm 0.30$ & $92.42 \pm 0.30$ & $92.53 \pm 0.14$ & $92.59 \pm 0.14$ & $92.75 \pm 0.04$\\
final  & $ 99.18 \pm 0.22$ & $99.28 \pm 0.33$ & $99.59 \pm 0.14$ & $99.51 \pm 0.17$ & $99.46 \pm 0.10$ & $99.50 \pm 0.05$ \\
\hline
\BankAge, initial  & $93.70 \pm 0.56$ & $93.35 \pm 0.41$ & $92.95 \pm 0.25$ & $93.28 \pm 0.13$ & $93.36 \pm 0.12$ & $93.33 \pm 0.12$ \\
final  & $99.40 \pm 0.28$ & $99.40 \pm 0.51 $ & $99.61 \pm 0.13$ & $ 99.64 \pm 0.07 $ & $99.65 \pm 0.08$ & $99.59 \pm 0.06$\\
\hline
\end{tabular}
\end{table*}

\noindent \textbf{Datasets.} We use two datasets from the UCI data repository.%
\footnote{\url{archive.ics.uci.edu/ml/index.php}, Census: \url{archive.ics.uci.edu/ml/datasets/census+income}, Bank: \url{archive.ics.uci.edu/ml/datasets/Bank+Marketing} }
In each dataset, we use features with numerical values and leave out samples with empty entries. For for value, we use the Euclidean distance as the dissimilarity measure. For revenue, we set the similarity to be $\simi(i,j) = \frac{1}{1+\dist(i,j)}$ where $\dist(i,j)$ is the Euclidean distance. We pick two different protected features for both datasets, resulting in four datasets in total (See Table \ref{table:datasets} for details).
\begin{itemize}[nosep]
\item \emph{Census} %
dataset: We choose \emph{gender} and \emph{race} to be the protected feature and call the resulting datasets \CensusGender and \CensusRace. 
\item \emph{Bank} %
dataset: We choose \emph{marital status} and \emph{age} to be the protected features and call the resulting datasets \BankMarriage and \BankAge.
\end{itemize}

In this section, unless otherwise specified, we report results only for the value objective. Results for the revenue objective are qualitatively similar and are omitted here.  We do not evaluate our algorithm
for the cost objective since it is currently only of theoretical interest.

We sub-sample points of two colors from the original data set proportionally, while approximately retaining the original color balance. The sample sizes used are 
$100 \times 2^i, i = 0, \ldots, 8$.  On each, we do $5$ experiments and report the average results. We set $\epsilon$ in Algorithm~\ref{alg:eps_local_opt_decomp} to $0.1$ in all of the experiments.

\noindent \textbf{Implementation.} The code is available in the Supplementary Material. In the experiments, we use Algorithm~\ref{alg:eps_local_opt_decomp} for the fairlet decomposition phase, where the fairlet decomposition is initialized by randomly assigning red and blue points to each fairlet.  We apply the average-linkage algorithm to create a tree on the fairlets.  We further use average-linkage to create subtrees inside of each fairlet.   

  The algorithm selects a \emph{random} pair of blue or red points in different fairlets to swap, and checks if the swap sufficiently improves the objective. We do not run the algorithm until all the pairs are checked, rather the algorithm stops if it has made a $2n$ failed attempts to swap a random pair. As we obseve empirically, this does not have material effect on the quality of the overall solution.

\iffalse
\begin{table*}[ht]
\centering 
\caption{Impact of Algorithm~\ref{alg:eps_local_opt_decomp} on $ratio_{value}$ in Percentage, Mean $\pm$ Std. Dev}
\label{emp:ratiockmm}
\tiny
\begin{tabular}{lllllllll}
\hline
Samples & 100 & 200 & 400 & 800 & 1500 & 3200 & 6400 & 12800\\
\hline
Initial, CensusGender & 87.53 \% & 87.86\% & 88.17\% & 88.39\% & 88.27\% & 88.12\% & 88.00\% & 88.04\% \\
Final, CensusGender & 98.50 \% & 98.93 \% & 99.01 \% & 99.09 \% & 99.55 \%  & 99.64 \% & 99.20 \%  & 99.44\% \\
\hline
Initial, CensusRace & 83.90\% & 84.21\% & 84.49\% & 85.01\% & 85.00\% & 84.88\% & 84.84\% & 84.89\% \\
Final, CensusRace & 98.29 \% & 98.81 \% & 99.50 \% & 99.89 \% & 100.00 \% & 99.98 \% & 99.98 \% & 99.93 \% \\
\hline
Initial, BankMarriage & 91.26\% & 91.82\% & 92.47\% & 92.58\% & 92.42\% & 92.53\% & 92.59\% & 92.75\%\\
Final, BankMarriage & 97.78\% & 99.12\% & 99.18\% & 99.28\% & 99.59\% & 99.51\% & 99.46\% & 99.50\% \\
\hline
Initial, BankAge & 92.71\% & 93.01\% & 93.70\% & 93.35\% & 92.95\% & 93.28\% & 93.36\% & 93.33\%\\
Final, BankAge & 98.39\% & 98.93\% & 99.40\% & 99.40\% & 99.61\% & 99.64\% & 99.65\% & 99.59\%\\
\hline
\hline
\end{tabular}
\end{table*}

\fi

\begin{table*}[t!]\vspace{-0.2in}
\centering
\caption{Impact of Algorithm~\ref{alg:eps_local_opt_decomp} on $\ratio_{{\text{fairlets}}}$.}
\label{table:ratio_fairness}
\tiny
\begin{tabular}{r|llllllll}
\hline
Samples & 100 & 200 & 400 & 800 & 1600 & 3200 & 6400 & 12800\\
\hline
\CensusGender, initial & 2.5e-2 & 1.2e-2 & 6.2e-3 & 3.0e-3 & 1.5e-3 & 7.5e-4 & 3.8e-4 & 1.9e-4 \\
final & 4.9e-3 & 1.4e-3 & 6.9e-4 & 2.5e-4 & 8.5e-5 & 3.6e-5 & 1.8e-5 & 8.0e-6
\\
\hline
\CensusRace, initial & 6.6e-2 & 3.4e-2 & 1.7e-2 & 8.4e-3 & 4.2e-3 & 2.1e-3 & 1.1e-3 & 5.3e-4 \\
final & 2.5e-2 & 1.2e-2 & 6.2e-3 & 3.0e-3 & 1.5e-3 & 7.5e-4 & 3.8e-4 & 1.9e-5 \\
\hline
\BankMarriage, initial & 1.7e-2 & 8.2e-3 & 4.0e-3 & 2.0e-3 & 1.0e-3 & 5.0e-4 & 2.5e-4 & 1.3e-4\\
final & 5.9e-3 & 2.1e-3 & 9.3e-4 & 4.1e-4 & 1.3e-4 & 7.1e-5 & 3.3e-5 & 1.4e-5 \\
\hline
\BankAge, initial & 1.3e-2 & 7.4e-3 & 3.5e-3 & 1.9e-3 & 9.3e-4 & 4.7e-4 & 2.3e-4 & 1.2e-4\\
final & 5.0e-3 & 2.2e-3 & 7.0e-4 & 3.7e-4 & 1.3e-4 & 5.7e-5 & 3.0e-5 & 1.4e-5 \\
\hline
\end{tabular}
\end{table*}
\begin{figure*}[t!]
\centering
\begin{tabular}{ccc} 
\includegraphics[width=.3\linewidth]{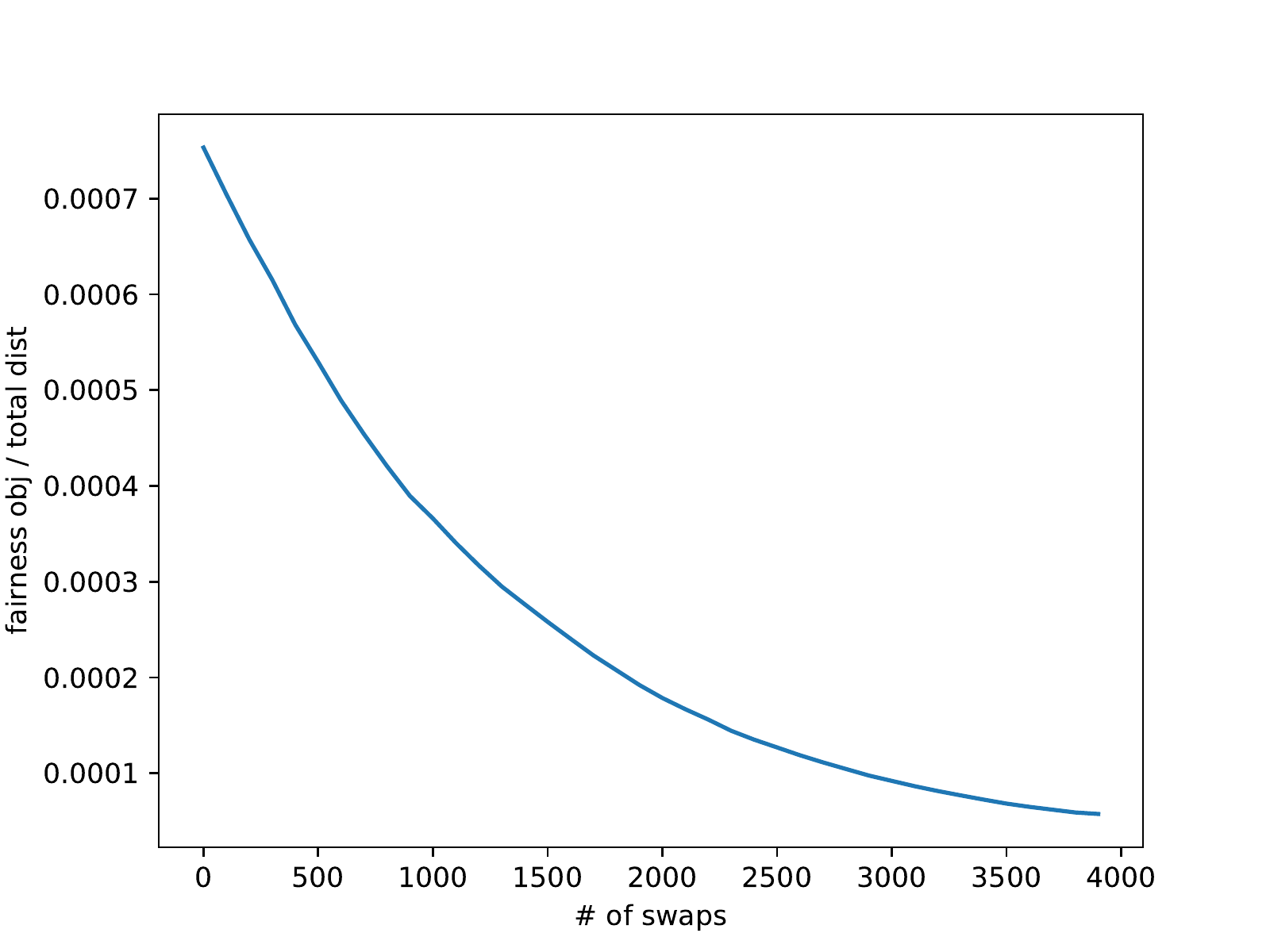}
&
\includegraphics[width=.3\linewidth]{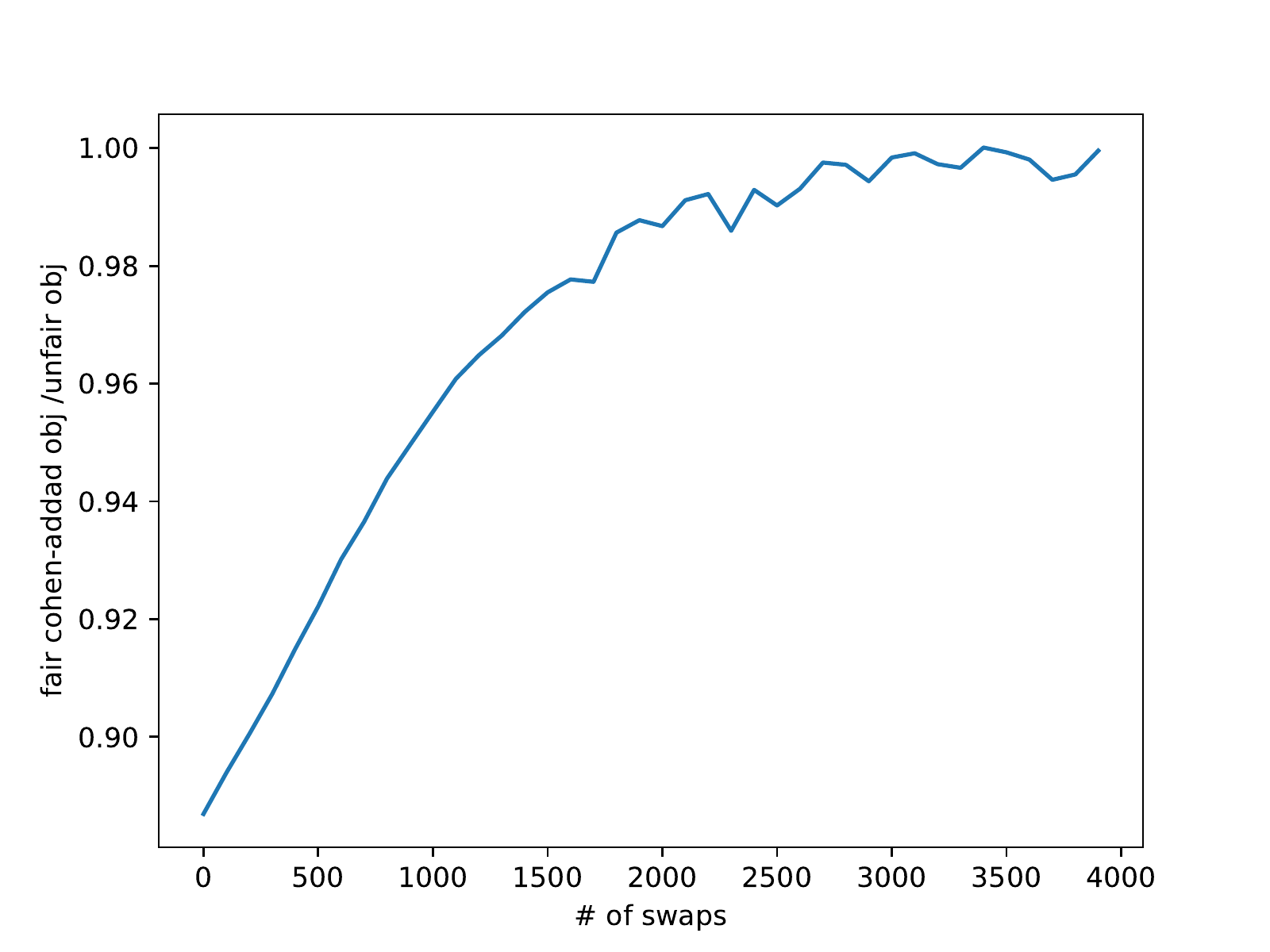}
&
\includegraphics[width=.3\linewidth]{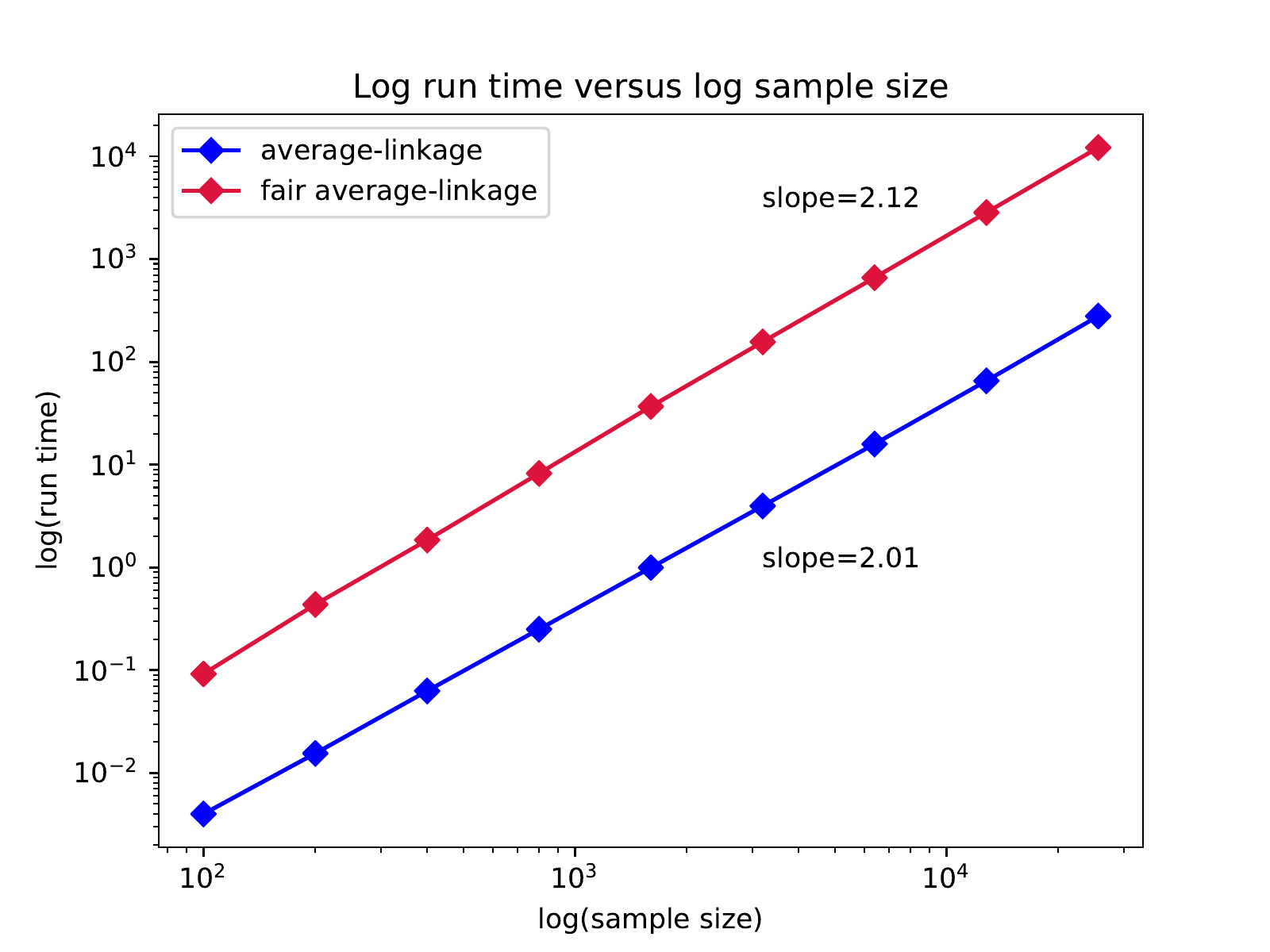}
\\
(i) & (ii) & (iii) \\
\end{tabular}
\caption{
(i) $\ratio_{\text{fairlets}}$, every $100$ swaps.  
(ii) $\ratio_{\myvalue}$, every $100$ swaps.
(iii) \CensusGender: running time vs sample size on a log-log scale.
}
\label{fig:combined}
\vspace{-0.15in}
\end{figure*}

\noindent \textbf{Metrics.}
 We present results for value here, the results for revenue are qualitatively similar.  In our experiments, we track the following quantities.  Let $G$ be the given input instance and let $T$ be the output of our fair hierarchical clustering algorithm.  We consider the following ratio $\ratio_\myvalue = \frac{\myvalue_G(T)}{ \myvalue_G(T')}$, where $T'$ is the tree obtained by the standard average-linkage algorithm. We  consider the fairlet objective function where $\mathcal{Y}$ is a fairlet decomposition. Let $\ratio_{{\text{fairlets}}} = \frac{\phi(\mathcal{Y})}{\dist(V)}$.

\noindent \textbf{Results.} 
Average-linkage algorithm always constructs unfair trees.  For each of the datasets, the algorithm results in monochromatic clusters at some level, strengthening the case for fair algorithms. %Results are qualitatively the same for all datasets.  

In Table~\ref{emp:ratiockmm}, we show for each dataset the $\ratio_\myvalue$ both at the time of initialization (Initial) and after usint the local search algorithm (Final).   We see the change in the ratio as the local search algorithm performs swaps. Fairness leads to almost no degradation in the objective value as the swaps increase. 
Table~\ref{table:ratio_fairness} shows the 
$\ratio_{{\text{fairlets}}}$ between the initial initialization and the final output fairlets.  As we see, Algorithm~\ref{alg:eps_local_opt_decomp} significantly improves the fairness of the initial random fairlet decomposition. 
\if 0
\begin{figure}[t!]
\centering
%\begin{minipage}{.5\textwidth}
  \centering
  \includegraphics[width=.7\linewidth]{fairness_obj_every_100_swaps_3200pts.pdf}
  \caption{$\ratio_{\text{fairlets}}$, every $100$ swaps.}
  \label{fig:test1}
  \vspace{-0.2in}
\end{figure}
\fi
\begin{table*}[t!] \vspace{-0.05in}
\centering 
\small
\caption{Clustering on fairlets found by local search vs. upper bound, at size 1600 (mean $\pm$ std. dev).}
\label{emp:size_1600_upper_bound}
\tiny
\begin{tabular}{r|llllllll}
\hline
Dataset & \CensusGender & \CensusRace & \BankMarriage & \BankAge \\
\hline
Revenue vs. upper bound& $81.89 \pm 0.40$ & $81.75 \pm 0.83$ & $61.53 \pm 0.37$ & $61.66 \pm 0.66$  \\
Value vs. upper bound & $84.31 \pm 0.15$ & $84.52 \pm 0.22$  & $89.17 \pm 0.29 $ & $88.81 \pm 0.18 $  \\
\hline
\end{tabular}
\vspace{-0.1in}
\end{table*}
The more the locally-optimal algorithm improves the objective value of \eqref{obj:fairlet_obj}, the better the tree's performance based on the fairlets.  Figures~\ref{fig:combined}(i) and \ref{fig:combined}(ii) show $\ratio_{\myvalue}$ and $\ratio_{{\text{fairlets}}}$ for every $100$ swaps in the execution of Algorithm~\ref{alg:eps_local_opt_decomp} on a subsample of size $3200$ from Census data set. The plots show that as the fairlet objective value decreases, the value\ objective of the resulting fair tree increases. Such correlation are found on subsamples of all sizes.

Now we compare the objective value of the algorithm with the upper bound on the optimum. We report the results for both the revenue and value objectives,  using fairlets obtained by local search, in Table~\ref{emp:size_1600_upper_bound}. On all datasets, we obtain ratios  significantly better than the theoretical worst case guarantee.   
\if 0
\begin{figure}[t!]
%\end{minipage}%
%\begin{minipage}{.5\textwidth}
  \centering
  \includegraphics[width=.7\linewidth]{cohen_addad_every_100_swaps_3200pts.pdf}
  \caption{$\ratio_{\myvalue}$, every $100$ swaps.}
  \label{fig:test2}
  \vspace{-0.2in}
%\end{minipage}
\end{figure}

\begin{figure}[t!]
\centering
  \includegraphics[width=.7\linewidth]{adult_runtime_b=1_r=3.pdf}
  \caption{\CensusGender: running time vs sample size on a log-log scale.}
  \label{fig:log_time_comparison}
  \vspace{-0.1in}
\end{figure}
\fi
In Figure~\ref{fig:combined}(iii), we show the average running time on Census data for both the original average-linkage and the fair average-linkage algorithms.  As the sample size grows, the running time scales almost as well as current implementations of average-linkage algorithm. Thus with a modest increase in time, we can obtain a fair hierarchical clustering under the value objective.

\vspace{-.1cm}
\section{Conclusions}
\vspace{-.1cm}
In this paper we extended the notion of fairness to the classical problem of hierarchical clustering under three different objectives (revenue, value, and cost).  Our results show that revenue and value are easy to optimize with fairness; while optimizing cost appears to be  more challenging. 

Our work raises several questions and research directions.  Can the approximations be improved? Can we find better upper and lower bounds for fair cost?  Are there other important fairness criteria?
%\newpage
%\section*{Broader Impact}

%Our work builds upon a long line of work of fairness in machine learning. See the excellent books by Kearns and Roth~\cite{rothbook}, and Barocas et al.~\cite{hardtbook} for a rich introduction to the field. 

%Our aim in this work is algorithmic in nature, finding near-optimal hierarchical clustering algorithms that attain certain fairness guarantees.  Since these methods are common unsupervised learning primitives, it is important to develop tools for practitioners to use. At the same time we remark that just because an algorithm is proven to be ``fair'' under some definition, does not mean it can be applied blindly. 

%As is now well known,~\cite{kleinbergimpossibility}, different fairness notions can be incompatible with each other. Moreover, fairness in machine learning is necessarily problem specific, and depends on the goals and the values of the person invoking the algorithm.  While these facts are well established in the research community, they are far from common knowledge outside of it. Thus work on algorithmic notions of fairness runs the risk of someone treating the results as a silver bullet, and eschewing the deeper analysis that is necessary in any real world application. 

\balance
\bibliography{references}
\bibliographystyle{plain}

\ifappendix
\newpage
\begin{appendix}
\section*{Appendix}
\newcommand{\Avg}{\mathit{Avg}}
\newcommand{\RED}{\mathit{red}}
\newcommand{\BLUE}{\mathit{blue}}

\section{Approximation algorithms for weighted hierarchical clustering}
\label{app:weighted_algo}

In this section we first prove that running constant-approximation algorithms on fairlets gives good solutions for value objective, and then give constant approximation algorithms for both revenue and value in weighted \hc problem, as is mentioned in Corollary~\ref{cor:revenue} and~\ref{cor:main}. That is, a weighted version of average-linkage, for both weighted revenue and value objective, and weighted ($\eps/n$)-locally densest cut algorithm, which works for weighted value objective. Both proofs are easily adapted from previous proofs in~\citet{cohenaddad} and~\citet{moseleywang}.

\subsection{Running constant-approximation algorithms on fairlets}

%\color{red}Redo to be aligned with the statement of the new thm \color{black}

In this section, we prove Theorem~\ref{thm:reduction}, which says if we run any $\beta$-approximation algorithm for the upper bound on weighted value on the fairlet decomposition, we get a fair tree with minimal loss in approximation ratio.  For the remainder of this section, fix any hierarchical clustering algorithm $A$ that is guaranteed on any \emph{weighted} input $(V,\dist,\weight)$ to construct a hierarchical clustering with objective value at least $\beta \weight(V) \dist(V)$ for the value objective on a weighted input.  Recall that we extended the value objective to a weighted variant in the Preliminaries Section and  $\weight(V) = \sum_{u\in V} \weight_u$. Our aim is to show that we can combine $A$ with the fairlet decomposition $\mathcal{Y}$ introduced in the prior section to get a fair hierarchical clustering that is a $\beta(1-\epsilon)$-approximation for the value objective, if $\phi(\mathcal{Y}) \leq \epsilon \dist(V)$. 

 In the following definition, we transform the point set to a new set of points that are weighted.  We will analyze $A$ on this new set of points. We then show how we can relate this to the objective value of the optimal tree on the original set of points.

\begin{definition}
Let  $\mathcal{Y} = \{Y_1, Y_2, \ldots \}$ be the fairlet decomposition for $V$ that is produced by the local search algorithm. Define $V(\mathcal{Y})$ as follows:
%\begin{enumerate}[(1)]
\begin{itemize}
    \item Each set $Y_i$  has a corresponding point $a_i$ in $V(\mathcal{Y})$.
    \item The weight $m_i$ of $a_i$ is set to be $|Y_i|$.
    \item For each partitions $Y_i,  Y_j$, where $i \neq j$ and $Y_i,Y_j \in \mathcal{Y}$, $\dist(a_i,a_j)= \dist(Y_i,Y_j)$.
\end{itemize}
%\end{enumerate}
\end{definition}

We begin by observing the objective value that $A$ receives on the instance $V(\mathcal{Y})$ is large compared to the weights in the original instance.

%\begin{claim}
\begin{theorem}\label{claim:rev}
On the instance  $V(\mathcal{Y})$ the algorithm $A$ has a total weighted objective of $\beta(1-\epsilon)  \cdot n \dist(V)$.
\end{theorem}
%\end{claim}

\begin{proof}
Notice that $m(V(\cy)) = |V| = n$. Consider the total  sum of all the distances in  $V(\mathcal{Y})$.  This is $\sum_{a_i,a_j \in V(\mathcal{Y})} \dist(a_i, a_j) = \sum_{Y_i,Y_j \in \mathcal{Y}} \dist(Y_i,Y_j) = \dist(V) - \phi(\mathcal{Y})$.  The upper bound on the optimal solution is $(\sum_{Y_i \in \mathcal{Y}}m_i)   (\dist(V) - \phi(\mathcal{Y})= n  (\dist(V) - \phi(\mathcal{Y}))$. Since $\phi(\mathcal{Y}) \leq \epsilon \dist(V)$, this upper bound is at least $ (1-\epsilon)n \dist(V)$.  Theorem \ref{thm:reduction} follows from the fact that the algorithm $A$ archives a weighted revenue at least a $\beta$ factor of the total weighted distances. 
\end{proof}

\subsection{Weighted hierarchical clustering: Constant-factor approximation}

For weighted \hc with positive integral weights, we define the weighted average-linkage algorithm for input $(V, \dist, \weight)$ and $(V, \simi, \weight)$. Define the \emph{average distance} to be $\Avg(A,B) = \frac{\dist(A,B)}{\weight(A) \weight(B)}$ for dissimilarity-based input, and $\Avg(A,B) = \frac{\simi(A,B)}{\weight(A) \weight(B)}$ for similarity-based input. In each iteration, weighted average-linkage seeks to merge the clusters which minimizes this value, if dissimilarity-based, and maximizes this value, if similarity-based.

\begin{lem} \label{lem:avlk_approx_weighted}
Weighted average-linkage is a $\frac{2}{3} (resp., \frac{1}{3})$ approximation for the upper bound on weighted value (resp., revenue) objective with positive, integral weights.
\end{lem}

\begin{proof}
We prove it for weighted value first. This is directly implied by the fact that average-linkage is $\frac{2}{3}$ approximation for unweighted value objective, as is proved in \citet{cohenaddad}. We have already seen in the last subsection that a unweighted input $V$ can be converted into weighted input $V(\cy)$. Vice versa, we can construct a weighted input $(V, \dist, \weight)$ into unweighted input with same upper bound for value objective. 

In weighted \hc we treat each point $p$ with integral weights as $\weight(p)$ duplicates of points with distance $0$ among themselves, let's call this set $S(p)$. For two weighted points $(p, \weight(p))$ and $(q, \weight(q))$, if $i \in S(p), j \in S(q)$, let $\dist(i,j) = \frac{\dist(p,q)}{\weight(p) \weight(q)}$. This \emph{unweighted} instance, composed of many duplicates, has the same upper bound as the weighted instance. Notice that running average-linkage on the unweighted instance will always choose to put all the duplicates $S(p)$ together first for each $p$, and then do \hc on top of the duplicates. Thus running average-linkage on the unweighted input gives a valid \hc tree for weighted input. Since unweighted value upper bound equals weighted value upper bound, the approximation ratio is the same.

Now we prove it for weighted revenue. In \citet{moseleywang}, average-linkage being $\frac{1}{3}$ approximation for unweighted revenue is proved by the following. Given any clustering $\mathcal{C}$, if average-linkage chooses to merge $A$ and $B$ in $\mathcal{C}$, we define a local revenue for this merge:
$$\text{merge-rev}(A,B) = \sum_{C \in \mathcal{C} \setminus \{A,B\}} |C||A||B| \Avg(A,B).$$
And correspondingly, a local cost:
\begin{align*}
\text{merge-cost}(A,B) =  \sum_{C \in \mathcal{C} \setminus \{A,B\}} (|B||A||C| \Avg(A,C) 
 + |A||B||C| \Avg(B,C)).
\end{align*}
Summing up the local revenue and cost over all merges gives the upper bound. \citet{moseleywang} used the property of average-linkage to prove that at every merge, $\text{merge-cost}(A,B) \leq 2\text{merge-rev}(A,B)$, which guarantees the total revenue, which is the summation of $\text{merge-rev}(A,B)$ over all merges, is at least $\frac{1}{3}$ of the upper bound. For the weighted case, we define 
\begin{align*}
\text{merge-rev}(A,B)
= \sum_{C \in \mathcal{C} \setminus \{A,B\}} \weight(C)\weight(A)\weight(B) \Avg(A,B).
\end{align*}
And 
\begin{align*}
\text{merge-cost}(A,B)
\sum_{C \in \mathcal{C} \setminus \{A,B\}} (\weight(B)\weight(A)\weight(C) \Avg(A,C)
+ \weight(A)\weight(B)\weight(C) \Avg(B,C)).
\end{align*}
And the rest of the proof works in the same way as in \citet{moseleywang}, proving weighted average-linkage to be $\frac{1}{3}$ for weighted revenue.
\end{proof}

Next we define the weighted $(\eps/n)$-locally-densest cut algorithm. The original algorithm, introduced in \citet{cohenaddad}, defines a cut to be $\frac{\dist(A,B)}{|A||B|}$. It starts with the original set as one cluster, at every step, it seeks the partition of the current set that locally maximizes this value, and thus constructing a tree from top to bottom. For the weighted input $(V, \dist, \weight)$, we define the cut to be  $\frac{\dist(A,B)}{\weight(A)\weight(B)}$, and let $n = \weight(V)$. For more description of the algorithm, see Algorithm 4 in Section 6.2 in \citet{cohenaddad}.

\begin{lem} \label{lem:densest_cut_approx_weighted}
Weighted $(\eps/n)$-locally-densest cut algorithm is a $\frac{2}{3}-\eps$ approximation for weighted value objective. 
\end{lem}

\begin{proof}
Just as in the average-linkage proof, we convert each weighted point $p$ into a set $S$ of $\weight(p)$ duplicates of $p$. Notice that the converted unweighted \hc input has the same upper bound as the weighted \hc input, and the $\eps/n$-locally-densest cut algorithm moves all the duplicate sets $S$ around in the unweighted input, instead of single points as in the original algorithm in \citet{cohenaddad}.

Focus on a split of cluster $A \cup B$ into $(A,B)$. Let $S$ be a duplicate set. $\forall S \subseteq A$, where $S$ is a set of duplicates, we must have
$$(1+\frac{\epsilon}{n})\frac{\dist(A,B)}{|A||B|} \geq \frac{\dist(A \setminus S, B \cup S)}{(|A| - |S|)(|B| + |S|)}.$$
Pick up a point $q \in S$, 
\begin{align*}
    & \qquad (1+\frac{\epsilon}{n})\dist(A,B)|S|(|A|-1)(|B|+1)\\
    & = (1+\frac{\epsilon}{n})\dist(A,B)(|A||B|+|A|-|B|-1)|S|\\
    & = (1+\frac{\epsilon}{n})\dist(A,B)(|A||B|+|A||S|-|B||S|-|S|) + (1+\frac{\epsilon}{n})\dist(A,B)(|A||B|)(|S|-1)\\
    & \geq (1+\frac{\epsilon}{n})\dist(A,B)(|A|-|S|)(|B|+|S|) + \dist(A,B)|A||B|(|S|-1)\\
    & \geq |A||B|\dist(A \setminus S, B \cup S) + \dist(A,B)|A||B|(|S|-1)\\
    & = |A||B| (\dist(A,B) + |S|\dist(q,A) - |S|\dist(q,B)) + |A||B|(|S|-1)\dist(A,B)\\
    & = |A||B||S|(\dist(A,B) + \dist(q,A) - \dist(q,B)).
\end{align*}
Rearrange the terms and we get the following inequality holds for any point $q \in A$:
$$\left( 1+\frac{\epsilon}{n} \right)\frac{\dist(A,B)}{|A||B|} \geq \frac{\dist(A,B) + \dist(q,A) - \dist(q,B)}{(|A|-1)(|B|+1)}.$$
The rest of the proof goes exactly the same as the proof in~\citet[Theorem 6.5]{cohenaddad}.
\end{proof}

\section{Proof of Theorem~\ref{thm:fair_moseley_wang}} \label{sec:thmproofsec}

%\begin{proofof}[Theorem~\ref{thm:fair_moseley_wang}]
\begin{proof}
Let $\mathcal{A}$ be the $\beta$-approximation algorithm to (\ref{eq:maxrevenue}).  For a given instance $G = (V, s)$, let $\cy = \{ Y_1, Y_2, \ldots \}$ be a fairlet decomposition of $V$; let $m_f = \max_{Y \in \mathcal{Y}} |Y|$.   Recall that $n = |V|$.  

We use $\cy$ to create a weighted instance $G_\cy = (\cy, \simi_\cy, \vweight_\cy)$.  
For $Y, Y' \in \cy$, we define
$\simi(Y, Y') = \sum_{i \in Y, j \in Y'} \simi(i, j)$ and we define $\vweight_\cy(Y) = |Y|$.

We run $\mathcal{A}$ on $G_\cy$ and let $T_{\cy}$ be the hierarchical clustering obtained by $\mathcal{A}$.  To extend this to a tree $T$ on $V$, we simply place all the points in each fairlet as leaves under the corresponding vertex in $T_\cy$. 

We argue that $\rev_G(T) \geq \beta\left(1-\frac{2\fsize}{n}\right)(n-2)\simi(V)$.

Since $\mathcal{A}$ obtains a $\beta$-approximation to hierarchical clustering on $G_\cy$, we have $\rev_{G_\cy}\big(T_\cy) \geq \beta \cdot \sum_{Y,  Y' \in \mathcal{Y}}\simi(Y, Y') (n- \vweight(Y) - \vweight(Y')). $

Notice the fact that, for any pair of points $u,v$ in the same fairlet $Y \in \cy$, the revenue they get in the tree $T$ is $(n-\vweight(Y))\simi(u,v)$. Then using $rev_G(T) = \sum_{Y \in \mathcal{Y}} (n-\vweight(Y)) \simi(Y) + \rev(T_\cy)$,
\begin{align*}
    \rev_G(T)   &\geq \sum_{Y \in \mathcal{Y}} \beta(n-\vweight(Y))\simi(Y)  + \beta \sum_{Y, Y' \in \mathcal{Y}} \simi(Y, Y')(n- \vweight(Y) - \vweight(Y')) \\
    &\geq \beta(n - 2\fsize)\left(\sum_{Y \in \cy}\simi(Y) + \sum_{Y, Y' \in \mathcal{Y}}\simi(Y, Y')\right)
    \geq \beta\left(1-\frac{2\fsize}{n}\right)(n-2)\simi(V).
\end{align*}
Thus the resulting tree $T$ is a $\beta\left(1-\frac{2\fsize}{n}\right)$-approximation of the upper bound.
\end{proof}

\section{Proofs for $(\eps/n)$-locally-optimal local search algorithm}
\label{app:algo_is_correct}
In this section, we prove that Algorithm \ref{alg:eps_local_opt_decomp} gives a good fairlet decomposition for the fairlet decomposition objective \ref{obj:fairlet_obj}, and that it has polynomial run time.

\subsection{Proof for a simplified version of Lemma~\ref{lem:algo_is_correct}}

In Subsection~\ref{app:multi_color}, we will prove Lemma~\ref{lem:algo_is_correct}. For now, we will consider a simpler version of Lemma~\ref{lem:algo_is_correct} in the  context of \citet{chierichetti}'s disparate impact problem, where we have red and blue points and strive to preserve their ratios in all clusters. Chierichetti et al.~\citet{chierichetti} provided a valid fairlet decomposition in this context, where each fairlet has at most $b$ blue points and $r$ red points. Before going deeper into the analysis,  we state the following useful proposition.

\begin{prop}\label{prop:rat}
 Let $\totr = |\RED(V)|$ be the total number of red points and $\totb=|\BLUE(V)|$ the number of blue points.  We have that, $\max \{\frac{r}{\totr}, \frac{b}{\totb}\} %= \frac{r}{\totr}  
 \leq \frac{2(b+r)}{n}$.
\end{prop}
\begin{proof}
Recall that $\mathit{balance}(V) = \frac{\totb}{\totr} \geq \frac{b}{r}$, and wlog $\totb \leq \totr$. Since the fractions are positive and  $\frac{\totb}{\totr} \geq \frac{b}{r}$ we know that $\frac{\totb}{\totb+ \totr} \geq \frac{b}{b+r}$. Since $\totb+\totr=n$ we conclude that $\totb \geq \frac{b}{b+r}n$.  Similarly, we conclude that $\frac{\totr}{\totb+ \totr} \leq \frac{r}{b+r}$.  Therefore $\totr \leq \frac{r}{b+r}n$.  

Thus, $\frac{r}{\totr} \geq \frac{b+r}{n} \geq \frac{b}{\totb}$. However, since $\totb \leq \totr$ and $\totb + \totr = n$, $\totr \geq \frac{1}{2}n$, $\frac{r}{\totr} \leq \frac{2r}{n} \leq \frac{2(b+r)}{n}$.
\end{proof}

Using this, we can define and prove the following lemma, which is a simplified version of Lemma~\ref{lem:algo_is_correct}.

\begin{lem} 
%\begin{theorem}
\label{lem:algo_is_correct_prop}
The fairlet decomposition $\mathcal{Y}$ computed by Algorithm  \ref{alg:eps_local_opt_decomp} has an objective value for (\ref{obj:fairlet_obj}) of at most $(1+\epsilon)\frac{2(b+r)}{n}\dist(V)$.
%\end{theorem}
\end{lem}

\begin{proof}

 Let $Y:V \mapsto \mathcal{Y}$ denote a mapping from a point in $V$ to the fairlet it belongs to. Let $d_R(X)=\sum_{u \in \RED(X)}\dist(u,X)$, and $d_B(X)=\sum_{v \in \BLUE(X)}d(v,X)$.  Naturally, $d_R(X)+d_B(X)=2d(X)$ for any set $X$. For a fairlet $Y_i \in \mathcal{Y}$, let $r_i$ and $b_i$ denote the number of red and blue points in $Y_i$.
 
We first bound the total number of intra-fairlet pairs. Let $x_i=|Y_i|$, we know that $0 \leq x_i \leq b+r$ and $\sum_{i}x_i=n$. The number of intra-fairlet pairs is at most $\sum_{i}x_i^2 \leq \sum_{i} (b+r) x_i = (b+r)n$.

The {\bf While} loop can end in two cases: 1) if $\mathcal{Y}$ is $(\epsilon/n)$-locally-optimal; 2) if $\sum_{Y_k \in \mathcal{Y}} d(Y_k) \leq \Delta$. Case 2 immediately implies the lemma, thus we focus on case 1. 
By definition of the algorithm, we know that for any pair $u \in Y(u)$ and $v \in Y(v)$ where $u,v$ have the same color and $Y(u) \neq Y(v)$ the swap does not increase objective value by a large amount. (The same trivially holds if the pair are in the same cluster.)  
 
\begin{align*}
\sum_{Y_k} \dist(Y_k)
& \leq (1+ \frac{\epsilon}{n})(\sum_{Y_k} \dist(Y_k) - \dist(u,Y(u))-\dist(v,Y(v)) 
 + \dist(u,Y(v)) +\dist(v,Y(u)) -2\dist(u,v))\\
&\leq (1+ \frac{\epsilon}{n})(\sum_{Y_k} \dist(Y_k) - \dist(u,Y(u))-\dist(v,Y(v))
+ \dist(u,Y(v)) +\dist(v,Y(u)) ).
\end{align*}
After moving terms and some simplification, we get the following inequality:
\begin{equation}     \label{neq:eps_local_swap1}
\begin{split}
    & \quad \dist(u,Y(u))+\dist(v,Y(v)) \\
    & \leq  \dist(u,Y(v)) + \dist(v,Y(u)) + \quad \frac{\epsilon/n}{1+\epsilon/n}\sum_{Y_k \in \mathcal{Y}}\dist(Y_k)\\
    & \leq \dist(u,Y(v)) + \dist(v,Y(u)) +  \frac{\epsilon}{n}\sum_{Y_k \in \mathcal{Y}}\dist(Y_k).
\end{split}
\end{equation}

Then we sum up (\ref{neq:eps_local_swap1}), $\dist(u,Y(u))+\dist(v,Y(v)) \leq \dist(u,Y(v)) + \dist(v,Y(u)) +  \frac{\epsilon}{n}\sum_{Y_k \in \mathcal{Y}}\dist(Y_k)
$, over every pair of points in $\RED(V)$ (even if they are in the same partition).  %Similarly we later consider $\BLUE(V)$.  For $\RED(V)$, we get
\begin{align*}
&\totr\sum_{Y_i}\dist_R(Y_i) \leq \Bigg(\sum_{Y_i}r_i \dist_R(Y_i) \Bigg)
+ \Bigg(  \sum_{u \in \RED(V)} \sum_{Y_i \neq Y(u)}r_i \dist(u,Y_i)\Bigg) + \totr^2 \frac{\epsilon}{n}\sum_{Y_i}\dist(Y_i).
\end{align*}
Divide both sides by $\totr$ and use the fact that  $r_i \leq r$ for all $Y_i$: 
\begin{align}
\sum_{Y_i}\dist_R(Y_i) \leq \left(\sum_{Y_i}\frac{r}{\totr} \dist_R(Y_i) \right) + \left(  \sum_{u \in \RED(V)} \sum_{Y_i \neq Y(u)}\frac{r}{\totr} \dist(u,Y_i)\right) + \frac{\totr\epsilon}{n} \sum_{Y_i}\dist(Y_i). \label{neq:red_swap_simplified_with_eps1}
\end{align}

For pairs of points in $\BLUE(V)$ we sum  (\ref{neq:eps_local_swap1}) to similarly obtain:
\begin{align}
\sum_{Y_i}\dist_B(Y_i) \leq \left(\sum_{Y_i}\frac{b}{\totb} \dist_B(Y_i) \right) + \left(  \sum_{v \in \BLUE(V)} \sum_{Y_i \neq Y(v)}\frac{b}{\totb} \dist(v,Y_i)\right) + \frac{\totb\epsilon}{n} \sum_{Y_i}\dist(Y_i). \label{neq:blue_swap_simplified_with_eps1}
\end{align}
Now we sum up \eqref{neq:red_swap_simplified_with_eps1} and \eqref{neq:blue_swap_simplified_with_eps1}. The LHS becomes:
$$\sum_{Y_i}(\dist_R(Y_i)+\dist_B(Y_i))=\sum_{Y_i}\sum_{u \in Y_i}\dist(u,Y_i)=2\sum_{Y_i}\dist(Y_i.)$$
For the RHS, the last term in  \eqref{neq:red_swap_simplified_with_eps1} and \eqref{neq:blue_swap_simplified_with_eps1} is $\frac{\eps(\totb+\totr)}{n} \sum_{Y_i}\dist(Y_i)=  \eps \sum_{Y_i}\dist(Y_i)$. 

The other terms give:
\begin{align*}
    &\quad \frac{r}{\totr}\sum_{Y_i} \dist_R(Y_i) + \frac{r}{\totr}\sum_{u \in \RED(V)}\sum_{Y_i \neq Y(u)}\dist(u,Y_i)
 + \frac{b}{\totb}\sum_{Y_i}\dist_B(Y_i) + \frac{b}{\totb}\sum_{v \in \BLUE(V)} \sum_{Y_i \neq Y(v)}\dist(v,Y_i)\\
    &\leq \max\{\frac{r}{\totr},\frac{b}{\totb}\} \Bigg\{\sum_{Y_i}\left( \dist_R(Y_i) + \dist_B(Y_i)\right)
    + \sum_{u \in V} \sum_{Y_i \neq Y(u)}\dist(u,Y_i) \Bigg\}\\
    &= \max\{\frac{r}{\totr},\frac{b}{\totb}\} \Bigg\{ \sum_{Y_i}\sum_{u \in Y_i}\dist(u,Y_i)
     + \sum_{Y_i} \sum_{Y_j \neq Y_i} \dist(Y_i,Y_j) \Bigg\}\\
    &=2 \max\{\frac{r}{\totr},\frac{b}{\totb}\}\dist(V)\\
    &\leq \frac{4(b+r)}{n} \dist(V).
    \end{align*}
    The last inequality  follows from Proposition~\ref{prop:rat}.  All together, this proves that
   \begin{align*}
   2\sum_{Y_k}\dist(Y_k) \leq \frac{4(b+r)}{n} \dist(V) + \eps\sum_{Y_k} \dist(Y_k).
    \end{align*}
Then, $\frac{\sum_{Y_k}\dist(Y_k)}{\dist(V)} \leq \frac{2(b+r)}{n}\cdot \frac{1}{1-\epsilon/2} \leq (1+\epsilon)\frac{2(b+r)}{n}$. The final step follows from the fact that $(1+\epsilon)(1-\epsilon/2)=1+\frac{\epsilon}{2}(1-\epsilon) \geq 1$.  This proves the lemma.
\end{proof}

\subsection{Proof for the generalized Lemma \ref{lem:algo_is_correct}}
\label{app:multi_color}

Next, we prove Lemma~\ref{lem:algo_is_correct} for the more generalized definition of fairness, which is $\alpha$-capped fairness. 

\begin{proofof}[Lemma~\ref{lem:algo_is_correct}]
The proof follows the same logic as in the two-color case: we first use the $(\epsilon/n)$-local optimality of the solution, and sum up the inequality over all pairs of points with the same color.

Let $Y:V \mapsto \mathcal{Y}$ denote a mapping from a point in $V$ to the fairlet it belongs to. Let $R_i(X)$ be the set of $R_i$ colored points in a set $X$. Let $\dist_{R_i}(X)=\sum_{u \in R_i(X)}\dist(u,X)$.  Naturally, $\sum_{i}\dist_{R_i}(x)=2\dist(X)$ for any set $X$ since the weight for every pair of points is repeated twice.

The {\bf While} loop can end in two cases: 1) if $\mathcal{Y}$ is $(\epsilon/n)$-locally-optimal; 2) if $\sum_{Y_k \in \mathcal{Y}} \dist(Y_k) \leq \Delta$. Case 2 immediately implies the lemma, thus we focus on case 1.  

By definition of the algorithm, we know that for any pair $u \in Y(u)$ and $v \in Y(v)$ where $u,v$ have the same color and $Y(u) \neq Y(v)$ the swap does not increase objective value by a large amount. (The same trivially holds if the pair are in the same cluster.)  We get the following inequality as in the two color case:
\begin{equation}
\dist(u,Y(u))+\dist(v,Y(v)) 
\leq \dist(u,Y(v)) + \dist(v,Y(u)) +  \frac{\epsilon}{n}\sum_{Y_k \in \mathcal{Y}}\dist(Y_k).
\end{equation}

For any color $R_i$, we sum it over every pair of points in $R_i(V)$ (even if they are in the same partition).  
\begin{align*}
n_i\sum_{Y_k}\dist_{R_i}(Y_k) \leq \Bigg(\sum_{Y_k}r_{ik} \dist_{R_i}(Y_k) \Bigg)
+ \Bigg(  \sum_{u \in R_i(V)} \sum_{Y_k \neq Y(u)}r_{ik} \dist(u,Y_k)\Bigg) + n_i^2 \frac{\epsilon}{n}\sum_{Y_k}\dist(Y_k).
\end{align*}
Divide both sides by $n_i$ and we get: 
\begin{align}
\sum_{Y_k}\dist_{R_i}(Y_k) \leq \left(\sum_{Y_k}\frac{r_{ik}}{n_i} \dist_{R_i}(Y_k) \right) + \left(  \sum_{u \in R_i(V)} \sum_{Y_k \neq Y(u)}\frac{r_{ik}}{n_i} \dist(u,Y_k)\right) + \frac{n_i\epsilon}{n} \sum_{Y_k}\dist(Y_k).
\end{align}

Now we sum up this inequality over all colors $R_i$. The LHS becomes:
$$\sum_{Y_k} \sum_{i}\dist_{R_i}(Y_k)=\sum_{Y_k}\sum_{u \in Y_k}\dist(u,Y_k)=2\sum_{Y_k}\dist(Y_k).$$
For the RHS, the last term sums up to $ \frac{\eps(\sum_{i}n_i)}{n} \sum_{Y_k}\dist(Y_k)=  \eps \sum_{Y_k}\dist(Y_k)$.  Using the fact that $\frac{r_{ik}}{n_i} \leq \max_{i,k}\frac{r_{ik}}{n_i}$, the other terms sum up to :
\begin{align*}
    &\quad \sum_{i} \sum_{Y_k} \frac{r_{ik}}{n_i}\dist_{R_i}(Y_k) + \sum_{i} \sum_{u \in R_i(V)}\sum_{Y_k \neq Y(u)} \frac{r_{ik}}{n_i} \dist(u,Y_k) \\
    &\leq \max_{i,k}\frac{r_{ik}}{n_i} \Bigg\{\sum_{Y_k}\sum_{i} \dist_{R_i}(Y_i) + \sum_{u \in V} \sum_{Y_k \neq Y(u)}\dist(u,Y_k) \Bigg\}\\
    &= \max_{i,k}\frac{r_{ik}}{n_i} \Bigg\{ \sum_{Y_k}\sum_{u \in Y_k}\dist(u,Y_k) + \sum_{Y_k} \sum_{Y_j \neq Y_k} \dist(Y_j,Y_k) \Bigg\}\\
    &=2 \max_{i,k}\frac{r_{ik}}{n_i} \cdot \dist(V).
    \end{align*}
Therefore, putting LHS and RHS together, we get
   \begin{align*}
   2\sum_{Y_k}\dist(Y_k) \leq 2\max_{i,k}\frac{r_{ik}}{n_i} \dist(V) + \eps\sum_{Y_k} \dist(Y_k).
    \end{align*}
Then, $\frac{\sum_{Y_k}\dist(Y_k)}{\dist(V)} \leq \max_{i,k}\frac{r_{ik}}{n_i} \cdot \frac{1}{1-\epsilon/2} \leq (1+\epsilon)\cdot \max_{i,k}\frac{r_{ik}}{n_i}$. The final step follows from the fact that $(1+\epsilon)(1-\epsilon/2)=1+\frac{\epsilon}{2}(1-\epsilon) \geq 1$.

\end{proofof}
In the two-color case, the ratio $\max_{i,k}\frac{r_{ik}}{n_i}$ becomes $\max \{\frac{r}{r_t}, \frac{b}{b_t}\}$, which can be further bounded by $\frac{2(b+r)}{n}$ (see Proposition~\ref{prop:rat}). If there exists a caplet decomposition such that $\max_{i,k}\frac{r_{ik}}{n_i} = o(1)$, Lemma \ref{lem:algo_is_correct} implies we can build a fair \hc tree with $o(1)$ loss in approximation ratio for value objective.

%In the special case where $\alpha = \frac{1}{c}$ for every color, that is, every color has equal representation in every set, it is easy to see that for any $i$, $\frac{r_{ik}}{n_i} = \frac{x_i}{n}$. Notice that, if $c=o(n)$, we can get fairlet decomposition with very small sizes, so $\frac{x_i}{n} \rightarrow 0$ if $n \rightarrow +\infty$.

Assuming for all color class $R_i$, $n_i \rightarrow +\infty$ as $n \rightarrow +\infty$, here we give a possible caplet decomposition for $\alpha = \frac{1}{t}(t <= c)$ with size $O(t)$ for positive integer $t$, thus guaranteeing $\max_{i,k}\frac{r_{ik}}{n_i} = o(1)$ for any $i$.

\begin{lem}\label{lem:fairletpartition}
For any set $P$ of size $p$ that satisfies fairness constraint with $\alpha = 1/t$, there exists a partition of $P$ into sets $(P_1, P_2, \ldots)$ where each $P_i$ satisfies the fairness constraint and $t \leq |P_i| < 2t$. 
\end{lem}
\begin{proof}
Let $p = m \times t + r$  with $ 0 \leq r  < t$, then the fairness constraints ensures that there are at most $m$ elements of each color. Consider partitioning obtained through the following process: consider an ordering of elements where points of the same color are in consecutive places, assign points to sets $P_1, P_2, \ldots, P_m$ in a round robin fashion. So each set $P_i$ gets at least $t$ elements and at most $t + r < 2t $ elements assigned to it. Since there are at most $m$ elements of each color, each set gets at most one point of any color and hence all sets satisfy the fairness constraint as $1 \leq \frac{1}{t} \cdot |P_i|$. 
\end{proof}

\subsection{Proof for the running time of $(\eps/n)$-locally-optimal fairlet decomposition algorithm}

\begin{proofof}[Lemma \ref{lem:algo_is_fast1}]
Notice that finding the maximum pairwise distance takes $O(n^2)$ time. Thus, we focus on analyzing the time spent on the {\bf While} loop.

Let $t$ be the total number of swaps. We argue that $t=\tilde{O}(n/\epsilon)$. If $t=0$ the conclusion trivially holds. Otherwise, consider the decomposition $\mathcal{Y}_{t-1}$ before the last swap. Since the {\bf While} loop does not terminate here, $\sum_{Y_k \in \mathcal{Y}_{t-1}}\dist(Y_k)\geq \Delta = \frac{b+r}{n} \dist_{max}$. However, at the beginning, we have $\sum_{Y_k \in \mathcal{Y}}\dist(Y_k) \leq (b+r)n \cdot \dist_{max}=n^2\Delta \leq n^2 \sum_{Y_k \in \mathcal{Y}_{t-1}}\dist(Y_k)$. Therefore, it takes at most $\log_{1+\epsilon/n}(n^2)=\tilde{O}(n/\epsilon)$ iterations to finish the {\bf While} loop.

It remains to discuss the running time of each iteration. We argue that there is a way to finish each iteration in $O(n^2)$ time. Before the {\bf While} loop, keep a record of $\dist(u,Y_i)$ for each point $u$ and each fairlet $Y_i$. This takes $O(n^2)$ time. If we know $\dist(u,Y_i)$ and the objective value from the last iteration, in the current iteration, it takes $O(1)$ time to calculate the new objective value after each swap $(u,v)$, and there are at most $n^2$ such calculations, before the algorithm either finds a pair to swap, or determines that no such pair is left. After the swap, the update for all the $\dist(u,Y_i)$ data takes $O(n)$ time. In total, every iteration takes $O(n^2)$ time.

Therefore, Algorithm~\ref{alg:eps_local_opt_decomp} takes $\tilde{O}(n^3/\epsilon)$ time.
\end{proofof}

\section{Hardness of optimal fairlet decomposition}
\label{app:hardness}

Before proving Theorem~\ref{thm:np_hard_constant_approx}, we state that the PARTITION INTO TRIANGLES (PIT) problem is known to belong to the NP-complete class \citep{garey2002computers}, defined as follows. In the definition, we call a clique $k$-clique if it has $k$ nodes. A triangle is a $3$-clique.

\begin{definition}
\emph{PARTITION INTO TRIANGLES\\ (PIT).} Given graph $G=(V,E)$, where $V=3n$, determine if $V$ can be partitioned into $3$-element sets $S_1,S_2, \ldots,S_n$, such that each $S_i$ forms a triangle in $G$. 
\end{definition}

 The NP-hardness of PIT problem gives us a more general statement.
\begin{definition}
\emph{PARTITION INTO $k$-CLIQUES\\ (PIKC).}
For a fixed number $k$ treated as constant, given graph $G=(V,E)$, where $V=kn$, determine if $V$ can be partitioned into $k$-element sets $S_1,S_2,\ldots,S_n$, such that each $S_i$ forms a $k$-clique in $G$.
\end{definition}

\begin{lem} \label{lem:np_hard_pikc}
For a fixed constant $k \geq 3$, the PIKC problem is NP-hard.
\end{lem}
\begin{proof}
We reduce the PIKC problem from the PIT problem. For any graph $G=(V,E)$ given to the PIT problem where $|V|=3n$, construct another graph $G'=(V',E')$. Let $V'=V \cup C_1 \cup C_2 \cup \cdots \cup C_n$, where all the $C_i$'s are $(k-3)$-cliques, and there is no edge between any two cliques $C_i$ and $C_j$ where $i \neq j$. For any $C_i$, let all points in $C_i$ to be connected to all nodes in $V$.

Now let $G'$ be the input to PIKC problem. We prove that $G$ can be partitioned into triangles if and only if $G'$ can be partitioned into $k$-cliques. If $V$ has a triangle partition $V=\{S_1,\ldots,S_n\}$, then $V'=\{S_1 \cup C_1,\ldots,S_n \cup C_n\}$ is a $k$-clique partition. On the other hand, if $V'$ has a $k$-clique partition $V'=\{S_1',\ldots,S_n'\}$ then $C_1,\ldots,C_n$ must each belong to different $k$-cliques since they are not connected to each other. Without loss of generality we assume $C_i \subseteq S_i$, then $V=\{S_1'\setminus C_1,\ldots,S_n' \setminus C_n\}$ is a triangle partition.
\end{proof}

We are ready to prove the theorem.

\begin{proofof}[Theorem~\ref{thm:np_hard_constant_approx}]
We prove Theorem~\ref{thm:np_hard_constant_approx} by proving that for given $z \geq 4$, if there exists a $c$-approximation polynomial algorithm $\mathcal{A}$ for (\ref{obj:fairlet_obj}), it can be used to solve the PIKC problem where $k=z-1$ for any instance as well. This holds for any finite $c$.

Given any graph $G=(V,E)$ that is input to the PIKC problem, where $|V|=kn=(z-1)n$, let a set $V'$ with distances be constructed in the following way:
\begin{enumerate}
    \item $V'=V \cup \{C_1,\ldots,C_n\}$, where each $C_i$ is a singleton.
    \item Color the points in $V$ red, and color all the $C_i$'s blue.
    \item For a $e=(u,v)$, let $\dist(u,v)=0$, if it satisfies one of the three conditions: 1) $e \in E$. 2) $u, v \in C_i$ for some $C_i$. 3) one of $u,v$ is in $V$, while the other belong to some $C_i$.
    \item All other edges have distance $1$.
\end{enumerate}
Obviously the blue points make up a $\nicefrac{1}{z}$ fraction of the input so each fairlet should have exactly $1$ blue point and $z-1$ red points.

We claim that $G$ has a $k$-clique partition if and only if algorithm $\mathcal{A}$ gives a solution of $0$ for (\ref{obj:fairlet_obj}). The same argument as in the proof of Lemma~\ref{lem:np_hard_pikc} will show that $G$ has a $k$-clique partition if and only if the optimal solution to (\ref{obj:fairlet_obj}) is $0$. This is equal to algorithm $\mathcal{A}$ giving a solution of $0$ since otherwise the approximate is not bounded. 
\end{proofof}

\section{Optimizing cost with fairness}\label{sec:cost_appendix}

In this section, we present our fair hierarchical clustering algorithm that approximates
Dasgupta's cost function and satisfies Theorem~\ref{thm:cost}. Most of the proofs can be found in Section~\ref{sec:cost_proofs}. We consider the problem of equal representation, where vertices are red or blue and $\alpha=1/2$. From now on, whenever we use the word ``fair'', we are referring to this fairness constraint. Our algorithm also uses parameters $\costm$ and $\costl$ such that $n\geq \costm\costl$ and $\costm>\costl+108\costm^2/\costl^2$ for $n=|V|$, and leverages a $\beta$-approximation for cost and $\gamma_\costm$-approximation for minimum weighted bisection. We will assume these are fixed and use them throughout the section.

We will ultimately show that we can find a fair solution that is a sublinear approximation for the unfair optimum $T^*_{\unfair}$, which is a lower bound of the fair optimum.  Our main result is Theorem~\ref{thm:cost}, which is stated in the body of the paper.

The current best approximations described in Theorem~\ref{thm:cost} are $\gamma_\costm = O(\log^{3/2} n)$ by~\citet{feige} and $\beta = \sqrt{\log n}$ by both~\citet{dasgupta} and~\citet{charikar17}. If we set $\costm =
\sqrt{n} (\log^{3/4} n)$ and $\costl = n^{1/3} \sqrt{\log n}$, then we get Corollary~\ref{cor:cost_main}.

\begin{cor}\label{cor:cost_main}
Consider the equal representation problem with two colors. There is an $O\left(n^{5/ 6}\log^{5/4} n \right)$-approximate fair clustering under the cost objective.
\end{cor}

The algorithm will be centered around a single clustering, which we call $\mathcal{C}$, that is extracted from an unfair hierarchy. We will then adapt this to become a similar, fair clustering $\mathcal{C}'$. To formalize what $\mathcal{C}'$ must satisfy to be sufficiently ``similar'' to $\mathcal{C}$, we introduce the notion of a $\mathcal{C}$-good clustering. Note that this is not an intuitive set of properties, it is simply what $\mathcal{C}'$ must satisfy in order 
\begin{definition}[Good clustering]
\label{def:good}
Fix a clustering $\mathcal{C}$ whose cluster sizes are at most $\costm$. A fair clustering $\mathcal{C}'$ is $\mathcal{C}$-\emph{good} if it satisfies
the following two properties:
\begin{enumerate}[topsep=0pt,itemsep=0pt,partopsep=0pt,parsep=0pt]
\item For any cluster $C\in \mathcal{C}$, there is a cluster $C'\in \mathcal{C}'$ such that all but (at most) an $O(\costl\gamma_\costm/\costm + \costm\gamma_\costm/\costl^2)$-fraction of the weight of edges in $C$ is also in $C'$.
\item Any $C'\in\mathcal{C}'$ is not too much bigger, so $|C'| \leq 6\costm\costl$.
\end{enumerate}
\end{definition}
The hierarchy will consist of a $\mathcal{C}$-good (for a specifically chosen $\mathcal{C}$) clustering $\mathcal{C}'$ as its only nontrivial layer. 
\begin{lem}\label{lem:finalcluster}
Let $T$ be a $\beta$-approximation for cost and $\mathcal{C}$ be a maximal clustering in $T$ under the condition that all cluster sizes are at most $\costm$.  Then, a fair two-tiered hierarchy $T'$ whose first level consists of a $\mathcal{C}$-good clustering achieves an $O\left(\frac{n}{\costm} + \costm\costl + \frac{n\costl\gamma_\costm}{\costm} + \frac{n\costm\gamma_\costm}{\costl^2}\right)\beta$-approximation for cost.
\end{lem}

\begin{proof}
    Since $T$ is a $\beta$-approximation, we know that:
    \[\cost(T) \leq \beta\cost(T^*_{\unfair})\]
    
    We will then utilize a scheme to account for the cost contributed by each edge relative to their cost in $T$ in the hopes of extending it to $T^*_{\unfair}$. There are three different types of edges:
    
    \begin{enumerate}[topsep=0pt,itemsep=0pt,partopsep=0pt,parsep=0pt]
        \item An edge $e$ that is merged into a cluster of size $\costm$ or greater in $T$, thus contributing $\costm \cdot \simi(e)$ to the cost. At worst, this edge is merged in the top cluster in $T'$ to contribute $n \cdot \simi(e)$. Thus, the factor increase in the cost contributed by $e$ is $O(n/\costm)$. Then since the total contribution of all such edges in $T$ is at most $\cost(T)$, the total contribution of all such edges in $T'$ is at most $O(n/\costm) \cdot \cost(T)$.
        \item An edge $e$ that started in some cluster $C\in\mathcal{C}$ that does not remain in the corresponding cluster $C'$. We are given that the total weight removed from any such $C$ is an $O(\costl\gamma_\costm/\costm + \costm\gamma_\costm/\costl^2)$-fraction of the weight contained in $C$. If we sum across the weight in all clusters in $\mathcal{C}$, that is at most $\cost(T)$. So the total amount of weight moved is at most $O(\costl\gamma_\costm/\costm + \costm\gamma_\costm/\costl^2) \cdot \cost(T)$. These edges contributed at least $2\simi(e)$ in $T$ as the smallest possible cluster size is two. In $T'$, these may have been merged at the top of the cluster, for a maximum cost contribution of $n \cdot \simi(e)$. Therefore, the total cost across all such edges is increased by at most a factor of $n/2$, which gives a total cost of at most $O(n\costl\gamma_\costm/\costm + n\costm\gamma_\costm/\costl^2) \cdot \cost(T)$.
        \item An edge $e$ that starts in some cluster $C\in\mathcal{C}$ and remains in the corresponding $C'\in\mathcal{C}'$. Similarly, this must have contributed at least $2\simi(e)$ in $T$, but now we know that this edge is merged within $C'$ in $T'$, and that the size of $C'$ is $|C'| \leq 6\costm\costl$. Thus its contribution increases at most by a factor of $3\costm\costl$. By the same reasoning from the first edge type we discussed, all these edges total contribute at most a factor of $O(\costm\costl) \cdot \cost(T)$.
    \end{enumerate}
    We can then put a conservative bound by putting this all together.
    \begin{align*}
    \cost(T') \leq& O\left(\frac n\costm + \costm\costl +\frac{n\costl\gamma_\costm}{\costm} + \frac{n\costm\gamma_\costm}{\costl^2}\right)\cost(T).
    \end{align*}
    Finally, we know $T$ is a $\beta$-approximation for $T^*_{\unfair}$.
    \begin{align*}
    \cost(T') & \leq O\left(\frac n\costm + \costm\costl  +\frac{n\costl \gamma_
    \costm}{\costm}+  \frac{n\costm\gamma_\costm}{\costl^2}\right) \\ 
    & \quad\qquad \cdot \beta \cdot \cost(T^*_{\unfair}).
    \qedhere
    \end{align*}
    
\end{proof}
With this proof, the only thing left to do is find a $\mathcal{C}$-good clustering $\mathcal{C}'$ (Definition~\ref{def:good}). Specifically, using the clustering $\mathcal{C}$ mentioned in Lemma~\ref{lem:finalcluster}, we would like to find a $\mathcal{C}$-good clustering $\mathcal{C}'$ using the following.
\begin{lem}\label{lem:buildcluster}
There is an algorithm that, given a clustering $\mathcal{C}$ with maximum cluster size $\costm$, creates a $\mathcal{C}$-good clustering.
\end{lem}
The proof is deferred to the Section~\ref{sec:cost_proofs}. With these two Lemmas, we can prove Theorem~\ref{thm:cost}.

\begin{proof}
Consider our graph $G$.  We first obtain a $\beta$-approximation for unfair cost, which yields a hierarchy tree $T$. Let $\mathcal{C}$ be the maximal clustering in $T$ under the constraint that the cluster sizes must not exceed $\costm$.  We then apply the algorithm from Lemma~\ref{lem:buildcluster} to get a $\mathcal{C}$-good clustering $\mathcal{C}'$. Construct $T'$ such that it has one layer that is $\mathcal{C}'$. Then we can apply the results from Lemma~\ref{lem:finalcluster} to get the desired approximation.
\end{proof}

From here, we will only provide a high-level description of the algorithm for Lemma~\ref{lem:buildcluster}. For precise details and proofs, see Section~\ref{sec:cost_proofs}. To start, we need to propose some terminology.
\begin{definition}[Red-blue matching]
A \textbf{red-blue matching} on a graph $G$ is a matching $M$ such that $M(u)=v$ implies $u$ and $v$ are different colors.
\end{definition}
Red-blue matchings are interesting because they help us ensure fairness. 
For instance, suppose $M$ is a red-blue matching that is also perfect (i.e., touches all nodes).  If the lowest level of a hierarchy consists of a clustering such that $v$ and $M(v)$ are in the same cluster for all $v$, then that level of the hierarchy is fair since there is a bijection between red and blue vertices within each cluster. When these clusters are merged up in the hierarchy, fairness is preserved.  

Our algorithm will modify an unfair clustering to be fair by combining clusters and moving a small number of vertices. To do this, we will use the following notion.  %Let $\mathcal{C} = \{ C_1, C_2, \ldots \}$ be a partition of the vertices of $G$.
\begin{definition}[Red-blue clustering graph]
Given a graph $G$ and a clustering $\mathcal{C}=\{C_1,\ldots,C_k\}$, we can construct a \textbf{red-blue clustering graph} $H_M = (V_M, E_M)$ that is associated with some red-blue matching $M$. Then $H_M$ is a graph where $V_M=\mathcal{C}$ and $(C_i,C_j)\in E_M$ if and only if there is a $v_i\in C_i$ and $M(v_i)=v_j\in C_j$.
\end{definition}
Basically, we create a graph of clusters, and there is an edge between two clusters if and only if there is at least one vertex in one cluster that is matched to some vertex in the other cluster. We now show that the red-blue clustering graph can be used to construct a fair clustering based on an unfair clustering.

\begin{prop}\label{prop:clusteringgraph}
Let $H_M$ be a red-blue clustering graph on a clustering $\mathcal{C}$ with a perfect red-blue matching $M$. Let $\mathcal{C}'$ be constructed by merging all the clusters in each component of $H_M$. Then $\mathcal{C}'$ is fair.
\end{prop}

\begin{proof}
Consider some $C\in \mathcal{C}'$. By construction, this must correspond to a connected component in $H_M$.  By definition of $H_M$, for any vertex $v\in C$, $M(v)\in C$.  That means $M$, restricted to $C$, defines a bijection between the red and blue nodes in $C$.  Therefore, $C$ has an equal number of red and blue vertices and hence is fair.
\end{proof}
We will start by extracting a clustering $\mathcal{C}$ from an unfair hierarchy $T$ that approximates cost. Then, we will construct a red-blue clustering graph $H_M$ with a perfect red-blue matching $M$. Then we can use the components of $H_M$ to define our first version of the clustering $\mathcal{C}'$. However, this requires a non-trivial way of moving vertices between clusters in $\mathcal{C}$.

We now give an overview of our algorithm in Steps (A)--(G). For a full description, see our pseudocode in Section~\ref{sec:costpseudo}.

(A) \textbf{Get an unfair approximation $T$}. We start by running a $\beta$-approximation for cost in the unfair setting. This gives us a tree $T$ such that $\cost(T) \leq \beta \cdot \cost (T^*_{\unfair})$.

(B) \textbf{Extract a $\costm$-maximal clustering}.  Given $T$, we find the maximal clustering $\mathcal{C}$ such that (i) every cluster in the clustering is of size at most $\costm$, and (ii) any cluster above these clusters in $T$ is of size more than $\costm$.

(C) \textbf{Combine clusters to be size $\costm$ to $3\costm$.} We will now slowly change $\mathcal{C}$ into $\mathcal{C'}$ during a number of steps. In the first step, we simply define $\mathcal{C}_0$ by merging small clusters $|C|\leq \costm$ until the merged size is between $\costm$ and $3\costm$.  Thus clusters in $\mathcal{C}$ are contained within clusters in $\mathcal{C}_0$, and all clusters are between size $\costm$ and $3\costm$.

(D) \textbf{Find cluster excesses.}  Next, we strive to make our clustering more fair. We do this by trying to find an underlying matching between red and blue vertices that agrees with $\mathcal{C}_0$ (matches are in the same cluster). If the matching were perfect, then the clusters in $\mathcal{C}_0$ would have equal red and blue representation. However, this is not guaranteed initially.  We start by conceptually matching as many red and blue vertices within clusters as we can. Note we do not actually create this matching; 
we just want to reserve the space for this matching to ensure fairness, but really some of these vertices may be moved later on. Then the remaining unmatched vertices in each cluster is either entirely red or entirely blue. We call this amount the \textit{excess} and the color the \textit{excess color}. We label each cluster with both of these.

(E) \textbf{Construct red-blue clustering graph}. Next, we would like to construct $H_M = (V_M, E_M)$, our red-blue clustering graph on $\mathcal{C}_0$.  Let $V_M = \mathcal{C}_0$.  In addition, for the within-cluster matchings mentioned in Step (D), let those matches be contained in $M$. With this start, we will do a matching process to simultaneously construct $E_M$ and the rest of $M$. Note the unmatched vertices are specifically the excess vertices in each cluster. We will match these with an iterative process given our parameter $\costl$:
    \begin{enumerate}[topsep=0pt,itemsep=0pt,partopsep=0pt,parsep=0pt]
        \item Select a vertex $C_i\in V_M$ with excess at least $\costl$ to start a new connected component in $H_M$. Without loss of generality, say its excess color is red.
        \item Find a vertex $C_j\in V_M$ whose excess color is blue and whose excess is at least $\costl$. Add $(C_i,C_j)$ to $E_M$.
        \item Say without loss of generality that the excess of $C_i$ is less than that of $C_j$. Then match all the excess in $C_i$ to vertices in the excess of $C_j$. Now $C_j$ has a smaller excess.
        \item If $C_j$ has an excess less than $\costl$ or $C_j$ is the $\costl$th cluster in this component, end this component. Start over at (1) with a new cluster.
        \item Otherwise, use $C_j$ as our reference and continue constructing this component at (2).
        \item Complete when there are no more clusters with over $\costl$ excess that are not in a component (or all remaining such clusters have the same excess color).
    \end{enumerate}
    We would like to construct $\mathcal{C}'$ by merging all clusters in each component. This would be fair if $M$ were a perfect matching, however this is not true yet.  In the next step, we handle this.

(F) \textbf{Fix unmatched vertices.}  We now want to match excess vertices that are unmatched. We do this by bringing vertices from other clusters into the clusters that have unmatched excess, starting with all small unmatched excess. Note that some clusters were never used in Step (E) because they had small excess to start. This means they had many internal red-blue matches. Remove $\costm^2/\costl^2$ of these and put them into clusters in need. For other vertices, we will later describe a process where $\costm/\costl$ of the clusters can contribute $108\costm^2/\costl^2$ vertices to account for unmatched excess. Thus clusters lose at most $108\costm^2/\costl^2$ vertices, and we account for all unmatched vertices. Call the new clustering $\mathcal{C}_1$. Now $M$ is perfect and $H_M$ is unchanged.

(G) \textbf{Define $\mathcal{C}'$.}  Finally, we create the clustering $\mathcal{C}'$ by merging the clusters in each component of $H_M$.  Note that Proposition~\ref{prop:clusteringgraph} assures $C'$ is fair. In addition, we will show that cluster sizes in $\mathcal{C}_1$ are at most $6\costm$, so $\mathcal{C}'$ has the desired upper bound of $6\costm\costl$ on cluster size. Finally, we removed at most $\costl+ \costm^2/\costl^2$ vertices from each cluster. This is the desired $\mathcal{C}$-good clustering. 

Further details and the proofs that the above sequence of steps achieve the desired approximation can be found in the next section.  While the approximation factor obtained is not as strong as the ones for revenue or value objectives with fairness, we believe cost is a much harder objective with fairness constraints.

\subsection{Proof of Theorem ~\ref{thm:cost}}\label{sec:cost_proofs}

This algorithm contains a number of components. We will discuss the claims made by the description step by step. In Step (A), we simply utilize any $\beta$-approximation for the unfair approximation. Step (B) is also quite simple. At this point, all that is left is to show how to find $\mathcal{C}'$, ie, prove Lemma~\ref{lem:buildcluster} (introduced in Section~\ref{sec:cost}). This occurs in the steps following Step (B). In Step (C), we apply our first changes to the starting clustering from $T$. We now prove that the cluster sizes can be enforced to be between $\costm$ and $3\costm$.

\begin{lem}\label{lem:maximalcluster}
Given a clustering $\mathcal{C}$, we can construct a clustering $\mathcal{C}_0$, where each $C\in \mathcal{C}_0$ is a union of clusters in $\mathcal{C}$ and $\costm\leq |C|< 3\costm$.
\end{lem}

\begin{proof}
We iterate over all clusters in $\mathcal{C}$ whose size are less than $\costm$ and continually merge them until we create a cluster of size $\geq \costm$. Note that since the last two clusters we merged were of size $<\costm$, this cluster is of size $\costm\leq|C|< 2\costm$. We then stop this cluster and continue merging the rest of the clusters. At the end, if we are left with a single cluster of size $<\costm$, we simply merge this with any other cluster, which will then be of size $\costm \leq |C| < 3\costm$. 
\end{proof}

Step (D) describes a rather simple process. All we have to do in each cluster is count the amount of each color in each cluster, find which is more, and also compute the difference. No claims are made here.

Step (E) defines a more careful process. We describe this process and its results here.

\begin{lem}\label{lem:hgraph}
There is an algorithm that, given a clustering $\mathcal{C}_0$ with  $\costm\leq |C| \leq 3\costm$ for $C\in\mathcal{C}_0$, can construct a red-blue clustering graph $H_M = (V_M, E_M)$ on $\mathcal{C}_0$ with underlying matching $M$ such that:
\begin{enumerate}
\item $H_M$ is a forest, and its max component size is $\costl$.
\item For every $(C_i, C_j)\in E_M$, there are at least $\costl$ matches between $C_i$ and $C_j$ in $M$. In other words, $|M(C_i)\cap C_j| \geq \costl$.
\item For most $C_i\in V_M$, at most $\costl$ vertices in $C_i$ are unmatched in $M$. The only exceptions to this rule are (1) exactly one cluster in every $\costl$-sized component in $H_M$, and (2) at most $n/2$ additional clusters.
\end{enumerate}
\end{lem}

\begin{proof}
We use precisely the process from Step 5. Let $V_M = \mathcal{C}_0$. $H_M$ will look like a bipartite graph with some entirely isolated nodes. We then try to construct components of $H_M$ one-by-one such that (1) the max component size is $\costl$, and (2) edges represent at least $\costl$ matches in $M$. 
\iffalse
This will go as follows:

\begin{enumerate}
        \item Select a vertex $C_i\in V_M$ with excess at least $\costl$ to start a new connected component in $H_M$. Without loss of generality, say its excess color is red.
        \item Find a vertex $C_j\in V_M$ whose excess color is blue and whose excess is at least $\costl$. Add $(C_i,C_j)$ to $E_M$.
        \item Say without loss of generality that the excess of $C_i$ is less than that of $C_j$. Then match all the excess in $C_i$ to vertices in the excess of $C_j$. Now $C_j$ has a smaller excess.
        \item If $C_j$ has an excess less than $\costl$ OR $C_j$ is the $\costl$th cluster in this component, end this component. Start over at (1) with a new cluster.
        \item Otherwise, use $C_j$ as our reference and continue constructing this component at (2).
        \item Complete when there are no more clusters with over $\costl$ excess that are not in a component.
\end{enumerate}
\fi

Let us show it satisfies the three conditions of the lemma. For condition 1, note that we will always halt component construction once it reaches size $\costl$. Thus no component can exceed size $\costl$. In addition, for every edge added to the graph, at least one of its endpoints now has small excess and will not be considered later in the program. Thus no cycles can be created, so it is a forest.

For condition 2, consider the construction of any edge $(C_i,C_j)\in E_M$. At this point, we only consider $C_i$ and $C_j$ to be clusters with different-color excess of at least $\costl$ each. In the next part of the algorithm, we match as much excess as we can between the two clusters. Therefore, there must be at least $\costl$ underlying matches.

Finally, condition 3 will be achieved by the completion condition. By the completion condition, there are no isolated vertices (besides possibly those leftover of the same excess color) that have over $\costl$ excess. Whenever we add a cluster to a component, either that cluster matches all of its excess, or the cluster it becomes adjacent to matches all of its excess. Therefore at any time, any component has at most one cluster with any excess at all. If the component is smaller than $\costl$ (and is not the final component), then that can only happen when in the final addition, both clusters end up with less than $\costl$ excess. Therefore, no cluster in this component can have less than $\costl$ excess. For an $\costl$-sized component, by the rule mentioned before, only one cluster can remain with $\costl$ excess. When the algorithm completes, we are left with a number of large-excess clusters with the same excess color, say red without loss of generality. Assume for contradiction there are more than $n/2$ such clusters, and so there is at least $n\ell/2$ . Since we started with half red and half blue vertices, the remaining excess in the rest of the clusters must match up with the large red excess. Thus the remaining at most $n/2$ clusters must have at least $n\ell/2$ blue excess, but this is only achievable if they have large excess left. This is a contradiction. Thus we satisfy condition 3.
\end{proof}

This concludes Step (E). In Step (F), we will transform the underlying clustering $\mathcal{C}_0$ such that we can achieve a perfect matching $M$. This will require removing a small number of vertices from some clusters in $\mathcal{C}_0$ and putting them in clusters that have unmatched vertices. This process will at most double cluster size.

\begin{lem}\label{lem:fixmatch}
There is an algorithm that, given a clustering $\mathcal{C}_0$ with $\costm\leq |C| \leq 3\costm$ for $C\in\mathcal{C}_0$, finds a clustering $\mathcal{C}_1$ and an underlying matching $M'$ such that:

\begin{enumerate}
\item There is a bijection between $\mathcal{C}_0$ and $\mathcal{C}_1$.
\item For any cluster $C_0\in\mathcal{C}_0$ and its corresponding $C_1\in\mathcal{C}_1$, $|C_0|-|C_1|\leq \costl + 108\costm^2/\costl^2$. This means that at most $\costl$ vertices are removed from $C_0$ in the construction of $C_1$.
\item For all $C_1\in \mathcal{C}_1$, $\costm-\costl-108\costm^2/\costl^2 \leq |C_1|\leq 6\costm$.
\item $M'$ is a perfect red-blue matching.
\item $H_M$ is a red-blue clustering graph of $\mathcal{C}_1$ with matching $M'$, perhaps with additional edges.
\end{enumerate}
\end{lem}

\begin{proof}
Use Lemma~\ref{lem:hgraph} to find the red-blue clustering graph $H_M$ and its corresponding graph $M$. Then we know that only one cluster in every $\costl$-sized component plus one other cluster can have a larger than $\costl$ excess. Since cluster sizes are at least $\costm$, $|V_M|\geq n/\costm$. This means that at most $n/(\costm\costl) + 1 = (n+\costm\costl)/(\costm\costl) \leq 2n/(\costm\costl)$ clusters need more than $\costl$ vertices. Since the excess is upper bounded by cluster size which is upper bounded by $3\costm$, this is at most $6n/\costl$ vertices in large excess that need matches.

We will start by removing all small excess vertices from clusters. This removes at most $\costl$ from any cluster. These vertices will then be placed in clusters with large excess of the right color. If we run out of large excess of the right color that needs matches, since the total amount of red and blue vertices is balanced, that means we can instead transfer the unmatched small excess red vertices to clusters with a small amount of unmatched blue vertices. In either case, this accounts for all the small unmatched excess. Now all we need to account for is at most $6n/\costl$ unmatched vertices in large excess clusters. At this point, note that the large excess should be balanced between red and blue. From now on, we will remove matches from within and between clusters to contribute to this excess. Since this always contributes the same amount of red and blue vertices by breaking matches, we do not have to worry about the balance of colors. We will describe how to distribute these contributions across a large number of clusters.

Consider vertices that correspond to clusters that (ignoring the matching $M$) started out with at most $\costl$ excess. So the non-excess portion, which is at least size $\costm-\costl$, is entirely matched with itself. We will simply remove $\costm^2/\costl^2$ of these matches to contribute.

Otherwise, we will consider vertices that started out with large excess. We must devise a clever way to break matches without breaking too many incident upon a single cluster. For every tree in $H_M$ (since $H_M$ is a forest by Lemma~\ref{lem:hgraph}), start at the root, and do a breadth-first search over all internal vertices. At any vertex we visit, break $\costl$ matches between it and its child (recall by by Lemma~\ref{lem:hgraph} that each edge in $H_M$ represents at least $\costl$ inter-cluster matches). Thus, each break contributes $2\costl$ vertices. We do this for every internal vertex. Since an edge represents at least $\costl$ matches and the max cluster size is at most $3\costm$, any vertex can have at most $3\costm/\costl$ children. Thus the fraction of vertices in $H_M$ that correspond to a contribution of $2\costl$ vertices is at least $\costl/(3\costm)$.

Clearly, the worst case is when all vertices in $H_M$ have large excess, as this means that fewer clusters are ensured to be able to contribute. By Lemma~\ref{lem:hgraph}, at least $n/2$ of these are a part of completed connected components (ie, of size $\ell$ or with each cluster having small remaining excess). So consider this case. Since $|V_M| \geq n/(3\costm)$, then this process yields $n\costl^2/(18\costm^2)$ vertices. To achieve $6n/\costl$ vertices, we must then run $108\costm^2/\costl^3$ iterations. If an edge no longer represents $\costl$ matches because of an earlier iteration, consider it a non-edge for the rest of the process. The only thing left to consider is if a cluster $C$ becomes isolated in $H_M$ during the process. We know $C$ began with at least $\costm$ vertices, and at most $\costl$ were removed by removing small excess. So as long as $\costm>\costl+108\costm^2/\costl^2$, we can remove the rest of the $108\costm^2/\costl^2$ vertices from the non-excess in $C$ (the rest must be non-excess) in the same way as vertices that were isolated in $H_M$ to start. Thus, we can account for the entire set of unmatched vertices without removing more than $108\costm^2/\costl^2$ vertices from any given cluster.

Now we consider the conditions. Condition 1 is obviously satisfied because we are just modifying clusters in $\mathcal{C}_0$, not removing them. The second condition is true because of our careful accounting scheme where we only remove $\costl+108\costm^2/\costl^2$ vertices per cluster. The same is true for the lower bound in condition 3. When we add them to new clusters, since we only add a vertex to match an unmatched vertex, we at most double cluster size. So the max cluster size is $6\costm$.

For the fourth condition, note that we explicitly executed this process until all unmatched vertices became matched, and any endpoint in a match we broke was used to create a new match. Thus the new matching, which we call $M'$, is perfect. It is still red-blue. Finally, note we did not create any matches between clusters. Therefore, no match in $M'$ can violate $H_M$. Thus condition 5 is met.
\end{proof}

Finally, we construct our final clustering in Step (G). However, to satisfy the qualities of Lemma \ref{lem:finalcluster}, we must first argue about the weight loss from each cluster.

\begin{lem}\label{lem:weightloss}
Consider any clustering $\mathcal{C}$ with cluster sizes between $\costm$ and $6\costm$. Say each cluster has a specified $r$ number of red vertices to remove and $b$ number of blue vertices to remove such that $r+b\leq x$ for some $x$, and $r$ (resp. $b$) is nonzero only if the number of red (resp. blue) vertices in the cluster is $O(n)$. Then we can remove the desired number of each color while removing at most an $O((x/\costm)\gamma_\costm)$ of the weight originally contained within the cluster.
\end{lem}

\begin{proof}
    Consider some cluster $C$ with parameters $r$ and $b$. We will focus first on removing red vertices. Let $C_r$ be the red vertex set in $C$. We create a graph $K$ corresponding to this cluster as follows. Let $b_0$ be a vertex representing all blue vertices from $C$, $b_0'$ be the ``complement'' vertex to $b_0$, and $R$ be a set of vertices $r_i$ corresponding to all red vertices in $C$. We also add a set of $2r - |C_r| + 2X$ dummy vertices (where $X$ is just some large value that makes it so $2r - |C_r| + X > 0$). $2r - |C_r| + X$ of the dummy vertices will be connected to $b_0$ with infinite edge weight (denote these $\delta_i$), the other $X$ will be connected to $b_0'$ with infinite edge weight (denote these $\delta_i'$). This will ensure that $b_0$ and $b_0'$ are in the same partitions as their corresponding dummies. Let $\simi_G$ and $\simi_K$ be the similarity function in the original graph and new graph respectively.
    
    \begin{align*}
    \simi_K(b_0, \delta_i) =& \infty
    \\\simi_K(b_0', \delta_i') =& \infty
    \end{align*}
    
    The blue vertex $b_0$ is also connected to all $r_i$ with the following weight (where $C_b$ is the set of blue vertices in $C$):
    
    \[\simi_K(b_0, r_i) = \sum_{b_j\in C_b} \simi_G(r_i, b_j) + \frac12\sum_{r_j\in R\setminus \{r_j\}} \simi_G(r_i, r_j)\]
    
    This edge represents the cumulative edge weight between $r_i$ and all blue vertices. The additional summation term, which contains the edge weights between $r_i$ and all other red vertices, is necessary to ensure our bisection cut will also contain the edge weights between two of the removed red vertices.
    
    Next, the edge weights between red vertices must contain the other portion of the corresponding edge weight in the original graph.
    
    \[\simi_K(r_i, r_j) = \frac12 \simi_G(r_i, r_j)\]
    
    Now, we note that there are a total of $2 - |C_r| + 2X + |C_r| = 2r + 2X$ vertices. So a bisection will partition the graph into vertex sets of size $r +X$. Obviously, in any approximation, $b_0$ must be grouped with all $\delta_i$ and $b_0'$ must be grouped with all $\delta_i'$. This means the $b_0$ partition must contain $|C_r|-r$ of the $R$ vertices, and the $b_0'$ partition must contain the other $r$. These $r$ vertices in the latter partition are the ones we select to move.
    
    Consider any set $S$ of $r$ red vertices in $K$. Then it is a valid bisection. We now show that the edge weight in the cut for this bisection is exactly the edge weight lost by removing $S$ from $K$. We can do this algebraically. We start by breaking down the weight of the cut into the weight between the red vertices in $S$ and $b_0$, and also the red vertices in $S$ and the red vertices not in $S$.
    
    \begin{align*}
        \simi_K&(S, V(K)\setminus S) \\=& \sum_{r_i\in S} \simi_K(b_0, r_i) + \sum_{r_i\in S, r_j\in R\setminus S} \simi_K(r_i, r_j)
        \\=&\sum_{r_i\in S} \left(\sum_{b_j\in B} \simi_G(r_i, b_j) + \frac12\sum_{r_j\in R\setminus \{r_j\}} \simi_G(r_i, r_j)\right)\\&+ \sum_{r_i\in S, r_j\in R\setminus S} \frac12 \simi_G(r_i, r_j)
        \\=&\sum_{r_i\in S} \left(\sum_{b_j\in B} \simi_G(r_i, b_j) + \frac12\sum_{r_j\in R\setminus \{r_j\}} \simi_G(r_i, r_j) \right.\\&\left.\quad\quad\quad+ \frac12\sum_{r_j\in R\setminus S}  \simi_G(r_i, r_j)\right)
    \end{align*}
    
    Notice that the two last summations have an overlap. They both contribute half the edge weight between $r_i$ and vertices in $R\setminus S$. Thus, these edges contribute their entire edge weight. All remaining vertices in $S\setminus \{r_i\}$ only contribute half their edge weight. We can then distribute the summation.
    
    \begin{align*}
        \simi_K&(S, V(K)\setminus S) \\=&\sum_{r_i\in S} \left(\sum_{b_j\in B} \simi_G(r_i, b_j) + \frac12\sum_{r_j\in S\setminus \{r_j\}} \simi_G(r_i, r_j)\right. \\&\left.\quad\quad\quad+ \sum_{r_j\in R\setminus S}  \simi_G(r_i, r_j)\right)
        \\=&\sum_{r_i\in S,b_j\in B} \simi_G(r_i, b_j) + \frac12\sum_{r_i\in S,r_j\in S\setminus \{r_j\}} \simi_G(r_i, r_j) 
        \\&+ \sum_{r_i\in S,r_j\in R\setminus S}  \simi_G(r_i, r_j)
    \end{align*}
    
    In the middle summation, note that every edge $e=(u,v)$ is counted twice when $r_i=u$ and $r_j=v$, and when $r_i=v$ and $r_j=u$. We can then rewrite this as:
    
    \begin{align*}
        \simi_K(S, V(K)\setminus S) =&\sum_{r_i\in S,b_j\in B} \simi_G(r_i, b_j) \\&+ \sum_{r_i,r_j\in S} \simi_G(r_i, r_j) \\&+ \sum_{r_i\in S,r_j\in R\setminus S}  \simi_G(r_i, r_j)
    \end{align*}
    
    When we remove $S$, we remove the connections between $S$ and blue vertices, the connections within $S$, and the connections between $S$ and red vertices not in $S$. This is precisely what this accounts for. Therefore, any bisection on $K$ directly corresponds to removing a vertex set $S$ of $r$ red vertices from $C$. If we have a $\gamma_\costm$-approximation for minimum weighted bisection, then, this yields a $\gamma_\costm$-approximation for the smallest loss we can achieve from removing $r$ red vertices.
    
    Now we must compare the optimal way to remove $r$ vertices to the total weight in a cluster. Let $\rho = |C_r|$ be the number of red vertices in a cluster. Then the total number of possible cuts to isolate $r$ red vertices is $\binom{\rho}{r}$. Let $\mathcal{S}$ be the set of all possible cuts to isolate $r$ red vertices. Then if we sum over the weight of all possible cuts (where weight here is the weight between the $r$ removed vertices and all vertices, including each other), that will sum over each red-red edge and blue-red edge multiple times. A red-red edge is counted if either of its endpoints is in $S\in\mathcal{S}$, and this happens $2\binom{\rho}{r-1} - \binom{R-1}{r-2} \leq 2\binom{\rho}{r-1}$ of the time. A blue-red edge is counted if its red endpoint is in $S$, which happens $\binom{\rho}{r-1} \leq 2\binom{\rho}{r-1}$. And of course, since no blue-blue edge is covered, each is covered under $2\binom{\rho}{r-1}$ times. Therefore, if we sum over all these cuts, we get at most $2\binom{\rho}{r-1}$ times the weight of all edges in $C$.

\[\sum_{S\in\mathcal{S}} \simi(S) \leq 2\binom{\rho}{r-1} \simi(C)\]

Let $OPT$ be the minimum possible cut. Now since there are $\binom{\rho}{r}$ cuts, we know the lefthand side here is bounded above by $\binom{\rho}{r}\simi(OPT)$.

\[\binom{\rho}{r} \simi(OPT)\leq 2\binom{\rho}{r-1} \simi(C)\]

We can now simplify.

\[\simi(OPT)\leq \frac{2r}{\rho} \simi(C)\]

But note we are given $\rho = O(\costm)$. So if we have a $\gamma_\costm$ approximation for the minimum bisection problem, this means we can find a way to remove $r$ vertices such that the removed weight is at most $O(r/\costm)\gamma_\costm$. We can do this again to get a bound on the removal of the blue vertices. This yields a total weight removal of $O(x/\costm)\gamma_\costm$.
\end{proof}

Finally, we can prove Lemma~\ref{lem:buildcluster}, which satisfies the conditions of Lemma \ref{lem:finalcluster}.

\begin{proof}
Start by running Lemma~\ref{lem:maximalcluster} on $\mathcal{C}$ to yield $\mathcal{C}_0$. Then we can apply Lemma~\ref{lem:fixmatch} to yield $\mathcal{C}_1$ with red-blue clustering graph $H_M$ and underlying perfect red-blue matching $M'$. We create $\mathcal{C}'$ by merging components in $H_M$ into clusters. Since the max component size is $\costl$ and the max cluster size in $\mathcal{C}_1$ is $6\costm$, then the max cluster size in $\mathcal{C}'$ is $6\costm\costl$. This satisfies condition 2 of being $\mathcal{C}$-good. In addition, it is fair by Proposition~\ref{prop:clusteringgraph}.

Finally, we utilize the fact that we only moved at most $\costl + 108\costm^2\costl^2$ vertices from any cluster, and note that we only move vertices of a certain color if we have $O(n)$ of that color in that cluster. Then by Lemma~\ref{lem:weightloss}, we know we lost at most $O(\costl\gamma_\costm/\costm + \costm\gamma_\costm/\costl^2)$ fraction of the weight from any cluster. This satisfies the second condition and therefore $\mathcal{C}'$ is $\mathcal{C}$-good.
\end{proof}

\section{Additional experimental results for revenue}

We have conducted experiments on the four datasets for revenue as well. The 
Table~\ref{emp:ratiorevenue} shows the ratio of fair tree built by using average-linkage on different fairlet decompositions. We run Algorithm \ref{alg:eps_local_opt_decomp} on the subsamples with Euclidean distances. Then we convert distances into similarity scores using transformation $\simi(i,j) = \frac{1}{1+\dist(i,j)}$. We test the performance of the initial random fairlet decomposition and final fairlet decomposition found by Algorithm~\ref{alg:eps_local_opt_decomp} for revenue objective using the converted similarity scores.

\begin{table*}[t!]
\centering 
\caption{Impact of different fairlet decomposition on ratio over original average-linkage in percentage (mean $\pm$ std. dev).}
\label{emp:ratiorevenue}
\tiny
\begin{tabular}{r|llllllll}
\hline
Samples & 100 & 200 & 400 & 800 & 1600 \\
\hline
\CensusGender, initial & $74.12 \pm 2.52$ & $76.16 \pm 3.42$ & $74.15 \pm 1.44$ & $70.17 \pm 1.01$ & $65.02 \pm 0.79$  \\
final  & $92.32 \pm 2.70$ & $95.75 \pm 0.74$ & $95.68 \pm 0.96 $ & $96.61 \pm 0.60 $ & $ 97.45 \pm 0.19 $ \\
\hline
\CensusRace, initial  & $65.67 \pm 7.53$ & $65.31 \pm 3.74$ & $61.97 \pm 2.50$ & $59.59 \pm 1.89$ & $56.91 \pm 0.82$  \\
final & $85.38 \pm 1.68$ & $92.98 \pm 1.89$ & $94.99 \pm 0.52$ & $96.86 \pm 0.85$ & $97.24 \pm 0.63$  \\
\hline
\BankMarriage, initial  & $75.19 \pm 2.53$ & $73.58 \pm 1.05$ & $74.03 \pm 1.33$ & $73.68 \pm 0.59$ & $72.94 \pm 0.63$ \\
final  & $ 93.88 \pm 2.16$ & $96.91 \pm 0.99$ & $96.82 \pm 0.36$ & $97.05 \pm 0.71$ & $97.81 \pm 0.49$  \\
\hline
\BankAge, initial & $77.48 \pm 1.45$ & $78.28 \pm 1.75 $ & $76.40 \pm 1.65$ & $ 75.95 \pm 0.77 $ & $75.33 \pm 0.28$   \\
final  & $91.26 \pm 2.66$ & $95.74 \pm 2.17 $ & $96.45 \pm 1.56$ & $ 97.31 \pm 1.94 $ & $97.84 \pm 0.92$ \\
\hline
\end{tabular}
\end{table*}

\section{Additional experimental results for multiple colors}\label{sec:multicolorexp}

We ran experiments with multiple colors and the results are analogous to those in the paper. We tested both Census and Bank datasets, with age as the protected feature. For both datasets we set 4 ranges of age to get 4 colors and used $\alpha=1/3$. We ran the fairlet decomposition in~\cite{ahmadian} and compare the fair hierarchical clustering's performance to that of average-linkage. The age ranges and the number of data points belonging to each color are reported in Table~\ref{table:datasets_multi_colors}. Colors are named $\{1,2,3,4\}$ descending with regard to the number of points of the color. The vanilla average-linkage has been found to be unfair: if we take the layer of clusters in the tree that is only one layer higher than the leaves, there is always one cluster with $\alpha > \frac{1}{3}$ for the definition of $\alpha$-capped fairness, showing the tree to be unfair.

\begin{table*}[ht] 
\centering 
\caption{Age ranges for all four colors for Census and Bank.} \label{table:datasets_multi_colors}
\tiny
\begin{tabular}{r|llll}
\hline
Dataset & Color 1 & Color 2 & Color 3 & Color 4 \\
\hline
\textsc{CensusMultiColor} & $(26, 38]:9796$ & $(38, 48]: 7131$ & $(48, +\infty): 6822$ & $(0, 26]: 6413$  \\
\textsc{BankMultiColor} & $(30, 38]:14845$ & $(38, 48]: 12148$ & $(48, +\infty): 11188$ & $(0, 30]: 7030$ \\
\hline
\end{tabular}
\end{table*}

As in the main body, in Table~\ref{emp:ratiockmm_multi_color}, we show for each dataset the $\ratio_\myvalue$ both at the time of initialization (Initial) and after using the local search algorithm (Final), where $\ratio_\myvalue$ is the ratio between the performance of the tree built on top of the fairlets and that of the tree directly built by average-linkage. 

\begin{table*}[ht]\vspace{-0.2in}
\centering 
\caption{Impact of Algorithm~\ref{alg:eps_local_opt_decomp} on $\ratio_{\myvalue}$ in percentage (mean $\pm$ std. dev).}
\label{emp:ratiockmm_multi_color}
\tiny
\begin{tabular}{r|llllllll}
\hline
Samples & 200 & 400 & 800 & 1600 & 3200 & 6400 \\
\hline
\textsc{CensusMultiColor}, initial  & $88.55 \pm 0.87$ & $88.74 \pm 0.46$ & $88.45 \pm 0.53$ & $88.68 \pm 0.22$ & $88.56 \pm 0.20$ & $88.46 \pm 0.30$\\
final   & $99.01 \pm 0.09$ & $99.41 \pm 0.57$ & $99.87 \pm 0.28$ & $99.80 \pm 0.27$ & $100.00 \pm 0.14$ & $99.88 \pm 0.30$\\
\hline
\textsc{BankMultiColor}, initial  & $90.98 \pm 1.17$ & $91.22 \pm 0.84$ & $91.87 \pm 0.32$ & $91.70 \pm 0.30$ & $91.70 \pm 0.18$ & $91.69 \pm 0.14$\\
final  & $98.78 \pm 0.22$ & $99.34 \pm 0.32$ & $99.48 \pm 0.16$ & $99.71 \pm 0.16$ & $99.80 \pm 0.08$ & $99.84 \pm 0.05$ \\
\hline
\end{tabular}
\end{table*}

Table~\ref{emp:ratiorevenue_multi_color} shows the performance of trees built by average-linkage based on different fairlets, for Revenue objective. As in the main body, the similarity score between any two points $i,j$ is $\simi(i,j) = \frac{1}{1+\dist(i,j)}$. The entries in the table are mean and standard deviation of ratios between the fair tree performance and the vanilla average-linkage tree performance. This ratio was calculated both at time of initialization (Initial) when the fairlets were randomly found, and after Algorithm \ref{alg:eps_local_opt_decomp} terminated (Final). 

\begin{table*}[ht]\vspace{-0.2in}
\centering 
\caption{Impact of Algorithm~\ref{alg:eps_local_opt_decomp} on revenue, in percentage (mean $\pm$ std. dev).}
\label{emp:ratiorevenue_multi_color}
\tiny
\begin{tabular}{r|llllllll}
\hline
Samples & 200 & 400 & 800 & 1600 & 3200 \\
\hline
\textsc{CensusMultiColor}, initial  & $75.76 \pm 2.86$ & $73.60 \pm 1.77$ & $69.77 \pm 0.56$ & $66.02 \pm 0.95$ & $61.94 \pm 0.61$\\
final   & $92.68 \pm 0.97$ & $94.66 \pm 1.66$ & $96.40 \pm 0.61$ & $97.09 \pm 0.60$ & $97.43 \pm 0.77$\\
\hline
\textsc{BankMultiColor}, initial  & $72.08 \pm 0.98$ & $70.96 \pm 0.69$ & $70.79 \pm 0.72$ & $70.77 \pm 0.49$ & $69.88 \pm 0.53$\\
final  & $94.99 \pm 0.79$ & $95.87 \pm 2.07$ & $97.19 \pm 0.81$ & $97.93 \pm 0.59$ & $98.43 \pm 0.14$\\
\hline
\end{tabular}
\end{table*}

Table \ref{emp:runtime_multi_color} shows the run time of Algorithm \ref{alg:eps_local_opt_decomp} with multiple colors.
\begin{table*}[ht]
\centering 
\caption{Average running time of Algorithm~\ref{alg:eps_local_opt_decomp} in seconds.}
\label{emp:runtime_multi_color}
\begin{tabular}{r|llllllll}
\hline
Samples & 200 & 400 & 800 & 1600 & 3200 & 6400 \\
\hline
\textsc{CensusMultiColor} & 0.43 & 1.76 & 7.34 & 35.22 & 152.71 & 803.59\\
\hline
\textsc{BankMultiColor} & 0.43 & 1.45 & 6.77 & 29.64 & 127.29 & 586.08 \\
\hline
\end{tabular}
\end{table*}

\iffalse
\begin{table}[h]
\centering
\begin{tabular}{l|ll}
 & Census & Bank \\
 \hline
Revenue, final (\%) & 96.11 & 99.53 \\
Revenue, random (\%) & 58.75 & 70.44 \\
\hline
Value, initial (\%) & 88.56 & 91.55 \\
Value, final (\%) & 99.99 &99.88 \\
\hline
Runtime(s) & 867 & 1803 \\
\hline
\end{tabular}
\caption{Results for multiple colors.}
\end{table}
\fi
\newcommand{\freeVerts}{\mathit{fV}}
\newcommand{\component}{\mathit{component}}

\section{Pseudocode for the cost objective}\label{sec:costpseudo}

\begin{algorithm}[tb]
   \caption{Fair hierarchical clustering for cost objective.}
   \label{alg:cost}
\begin{algorithmic}
   \STATE {\bfseries Input:} Graph $G $, edge weight $w:E\to\mathbb{R}$, color $c:V\to\{\text{red, blue}\}$, parameters $t$ and $\ell$
   \STATE
   \STATE \COMMENT{Step (A)}
   \STATE $T \gets \textsc{UnfairHC}(G,w)$
   \COMMENT{Blackbox unfair clustering that minimizes cost}
   \STATE
   \STATE \COMMENT{Step (B)}
   \STATE Let $\mathcal{C} \gets \emptyset$
   \STATE Do a \textsf{BFS} of $T$, placing visited cluster $C$ in $\mathcal{C}$ if $|C|\leq t$, and not proceeding to $C$'s children
   \STATE
   \STATE \COMMENT{Step (C)}
   \STATE $\mathcal{C}_0, C' \gets \emptyset$
   \FOR{$C$ {\bfseries in} $\mathcal{C}$}
   \STATE $C' \leftarrow C' \cup C$
   \IF{$|C'| \geq t$}
   \STATE Add $C'$ to $\mathcal{C}_0$
   \STATE Let $C' \gets \emptyset$
   \ENDIF
   \ENDFOR
   \STATE If $|C'| > 0$, merge $C'$ into some cluster in $\mathcal{C}_0$
   \STATE
   \STATE \COMMENT{Step (D)}
   \FOR{$C$ {\bfseries in} $\mathcal{C}_0$}
   \STATE Let $exc(C) \gets$ majority color in $C$
   \STATE Let $ex(C) \gets$ difference between majority and minority colors in $C$
   \ENDFOR
   \STATE
   \STATE \COMMENT{Step (E}
   \STATE $H_M \gets$ BuildClusteringGraph$(\mathcal{C}_0, ex, exc)$
   \STATE
   \STATE \COMMENT{Step (F)}
   \STATE $\freeVerts \gets$ FixUnmatchedVertices$(\mathcal{C}_0, H_M, ex, exc)$
   \STATE
      \STATE \COMMENT{Step (G)}
   \STATE $\mathcal{C}' \gets$ ConstructClustering$(\mathcal{C}_0, ex , exc, \freeVerts)$
   \RETURN $\mathcal{C}'$
\end{algorithmic}
\end{algorithm}

\begin{algorithm}[tb]
   \caption{BuildClusteringGraph $(\mathcal{C}_0, ex, exc)$}
   \label{alg:stepe}
\begin{algorithmic}
   %\STATE {\bfseries Input:} graph $G$, color function $c:V\to\{\text{RED,BLUE}\}$, parameter $\ell$, clustering $\mathcal{C}_0$, excess functions $ex:\mathcal{C}_0\to\mathbb{N}$ and $exc:\mathcal{C}_0\to\{\text{RED,BLUE}\}$
   \STATE $H_M \gets (V_M=\mathcal{C}_0, E_M=\emptyset)$
   %\STATE Init $V_M = \{C_i\in \mathcal{C}_0: ex(C_i) \geq \ell\}$
   %\STATE Init function $root:K_H\to V_M$ where if $k_{C_i} = \{C_i\}$ for $C_i\in V_M$, $root(k_{C_i})= C_i$
   \STATE Let $C_i \in V_M$ be any vertex
   \STATE Let $\ell \gets n^{1/3}\sqrt{\log n}$
   %\STATE Let $compSize = 1$
   \WHILE{$\exists$ an unvisited $C_j\in V_M$ such that $exc(C_j)\neq exc(C_i)$} 
   \STATE Add $(C_i,C_j)$ to $E_M$
   %\STATE $compSize += 1$
   \STATE Swap labels $C_i$ and $C_j$ if $ex(C_j) > ex(C_i)$
   \STATE Let $ex(C_i) \leftarrow ex(C_i) - ex(C_j)$
   \IF{$ex(C_i) < \ell$ or $|\component(C_i)| \geq \ell$}
   \STATE Reassign starting point $C_i$ to an unvisited vertex in $V_M$
   %\STATE Let $compSize = 1$
   \ENDIF
   \ENDWHILE
   \RETURN $H_M$
\end{algorithmic}
\end{algorithm}

\begin{algorithm}[tb]
   \caption{FixUnmatchedVertices$(\mathcal{C}_0, H_M, ex, exc)$}
   \label{alg:stepf}
\begin{algorithmic}
   %\STATE {\bfseries Input:} graph $H_M$ with components $K_H$, root function $root:K_H\to V_M$, color function $c:V\to\{\text{RED,BLUE}\}$, parameters $t$ and $\ell$, clustering $\mathcal{C}_0$, excess functions $ex:\mathcal{C}_0\to\mathbb{N}$ and $exc:\mathcal{C}_0\to\{\text{RED,BLUE}\}$
   \STATE Let $\ell \gets n^{1/3}\sqrt{\log n}$
   \FOR{$C\in\mathcal{C}_0\setminus V_M$}
   \STATE Let $\freeVerts(C,\text{red}), \freeVerts(C,\text{blue}) 
   \gets m^2/\ell^2$
   \ENDFOR
   \FOR{$i$ {\bfseries from} $1$ {\bfseries to} $108t^2/\ell^3$}
   \FOR{each $k$ component in $H_M$}
   \FOR{$p$ {\bfseries in} a \textsf{BFS} of $k$}
   \STATE Let $ch \gets$ some child of $p$
   \STATE $\freeVerts(p, exc(p)) \gets \freeVerts(p, exc(p)) + \ell$ 
   \STATE $ex(p) \gets ex(p) - \ell$
   \STATE $\freeVerts(ch, exc(ch)) \gets \freeVerts(ch, exc(ch)) + \ell$ 
   \STATE $ex(ch) \gets ex(ch) - \ell$
   \IF{\# matches between $p$ and $ch$ $< \ell$}
   \STATE Remove $(p,ch)$ from $E_M$
   \COMMENT {This creates a new component}
   \ENDIF
   \ENDFOR
   \ENDFOR
   \ENDFOR
   \RETURN $\freeVerts$
\end{algorithmic}
\end{algorithm}

\iffalse
\begin{algorithm}[tb]
   \caption{Break Components}
   \label{alg:stepff}
\begin{algorithmic}
   %\STATE {\bfseries Input:} component tree $K\in K_H$, vertex $p\in V(K)$, graph $H_M$ with components $K_H$, root function $root:K_H\to V(H_M)$, color function $c:V\to\{\text{RED,BLUE}\}$, parameter $\ell$, clustering $\mathcal{C}_0$, excess functions $ex:\mathcal{C}_0\to\mathbb{N}$ and $exc:\mathcal{C}_0\to\{\text{RED,BLUE}\}$
   \STATE Find some $c$ child of $p$
   \IF{$exc(p) == \text{RED}$}
   \STATE $freeRVerts(p) += \ell$
   \STATE $freeBVerts(c) += \ell$
   \ELSE 
   \STATE $freeRVerts(c) += \ell$
   \STATE $freeBVerts(p) += \ell$
   \ENDIF
   \STATE $ex(p) -= \ell$
   \STATE $ex(p) -= \ell$
   \IF{Less than $\ell$ matches between $p$ and $c$ remain}
   \STATE Remove $(p,c)$ from $E(H_M)$
   \STATE For new component $K'$, set $root(K') = c$
   \ENDIF
\end{algorithmic}
\end{algorithm}
\fi

\begin{algorithm}[tb]
   \caption{ConstructClustering$(\mathcal{C}_0, ex , exc, \freeVerts)$}
   \label{alg:stepg}
\begin{algorithmic}
   %\STATE {\bfseries Input:} graph $G$, graph $H_M$, color function $c:V\to\{\text{RED,BLUE}\}$, parameter $\ell$, clustering $\mathcal{C}_0$, excess functions $ex:\mathcal{C}_0\to\mathbb{N}$ and $exc:\mathcal{C}_0\to\{\text{RED,BLUE}\}$
   \STATE Let $\mathcal{C}', R \gets \emptyset$
   \FOR{$C$ in $\mathcal{C}_0$}
   \FOR{$c$ {\bfseries in} $\{\text{red, blue}\}$}
   \STATE Let $f = \freeVerts(C, c)$
   \STATE Let $C_f = \{v\in C:c(v)=c\}$
   \STATE Create the transformed graph $L$ from $C_f$ \COMMENT{Described in the proof of Lemma~\ref{lem:weightloss}}
   \STATE $C' \gets \textsc{MinWeightBisection}(L)$ \COMMENT{Blackbox, returns isolated $C_f$ vertices}
   \STATE $C \gets C\setminus C'$
   \STATE $R \gets R \cup C'$
   \STATE $ex(C) \gets ex(C) - |C'|$
   \ENDFOR
   \ENDFOR
   \FOR{$C\in \mathcal{C}_0$}
   \STATE Let $S\subset R$ such that $|S| = ex(C)$ with no vertices of color $exc(C)$
   \STATE $C = C\cup S$
   \STATE $R \gets R \setminus S$
   \STATE Add $C$ to $\mathcal{C}'$
   \ENDFOR
   \RETURN $\mathcal{C}'$
\end{algorithmic}
\end{algorithm}
\end{appendix}
\fi

\end{document}